\documentclass[12pt]{article}
\usepackage{amsmath,amssymb,amsthm}
\usepackage{graphics,epsfig}
\usepackage{hyperref}
\usepackage{natbib}
\usepackage{color}
\usepackage{graphicx}
\usepackage{caption}
\usepackage{subcaption}
\usepackage{float}
\usepackage{mathrsfs}
\usepackage{multirow}
\usepackage{comment}
\usepackage{setspace}
\usepackage{bbm}
\usepackage{bm}
\usepackage{amsfonts}
\usepackage{xcolor}
\usepackage{pslatex}
\usepackage{setspace}
\usepackage{booktabs}

\def \ra {\rightarrow}

\def \E {\mathbb{E}}

\def \be {\beta}

\newtheorem{definition}{\bf Definition}
\newtheorem{defn}[definition]{\bf Definition}

	\newtheorem{remark}{\bf Remark}
	\newtheorem{rmk}[remark]{\bf Remark}

	\newtheorem{theorem}{\bf Theorem}
	
	\newtheorem{prop}[theorem]{\bf Proposition}
	\newtheorem{lem}[theorem]{\bf Lemma}
	\newtheorem{cor}[theorem]{\bf Corollary}

        \newtheorem{as}[theorem]{\bf Assumption}


\begin{document}

  \title{\bfseries Impact of random monetary shock: a Keynesian case}

  \author{
Paramahansa Pramanik\footnote{Corresponding author, {\small\texttt{ppramanik@southalabama.edu}}}\; \footnote{Department of Mathematics and Statistics,  University of South Alabama, Mobile, AL, 36688,
United States.}
\and
Lambert Dong \footnotemark[2]
}

\date{\today}
\maketitle

\begin{abstract}
This study investigates the optimal strategy for a firm operating in a dynamic Keynesian market setting. The firm's objective function is optimized using the percent deviations from the symmetric equilibrium of both its own price and the aggregate consumer price index (CPI) as state variables, with the strategy in response to random monetary shocks acting as the control variable. Building on the Calvo framework, we adopt a mean field approach to derive an analytic expression for the firm’s optimal strategy. Our theoretical results show that greater volatility leads to a decrease in the optimal strategy. To assess the practical relevance of our model, we apply it to four leading consumer goods firms. Empirical analysis suggests that the observed decline in strategies under uncertainty is significantly steeper than what the model predicts, underscoring the substantial influence of market volatility.
\end{abstract}


\section{Introduction:}

Although recent advancements have deepened our understanding of general equilibrium models with price rigidity, many existing frameworks still sacrifice realism for analytical ease, often neglecting how firms' pricing strategies influence one another. These interactions are not merely theoretical curiosities, they are essential to explaining how nominal shocks can produce outsized effects on the real economy, as emphasized by \cite{klenow2016real} and \cite{nakamura2018high}. Nonetheless, a large share of the literature either sidesteps these complexities by relying on computational methods \citep{nakamura2018high, klenow2016real, mongey2021market}, oversimplifies timing mechanisms \citep{wang2022dynamic}, or omits firm-specific randomness altogether \citep{caplin1991state}. As a result, the mutual feedback between individual firms’ behavior and overall economic dynamics is often underrepresented in existing theoretical models.

To address these gaps, this study presents a novel analytical model that embeds both strategic pricing complementarities and firm-level shocks within a state-dependent sticky price setting. By building on the driftless Calvo pricing structure and incorporating a mean field approach, the model yields a closed-form characterization of optimal firm behavior in response to macroeconomic conditions. This setup allows firms’ decisions to be shaped by, and to simultaneously influence-aggregate variables, capturing the reciprocal relationship between microeconomic heterogeneity and macro-level fluctuations. The resulting framework not only enhances our understanding of how sticky prices behave in general equilibrium but also establishes a versatile analytical foundation for future work in macroeconomic modeling, including monetary policy analysis and business cycle research \citep{alvarez2023price}.

Analyzing strategic complementarities within a comprehensive equilibrium framework introduces significant challenges: individual decisions depend on aggregate conditions, which themselves are shaped by those decisions. This interdependence creates a fixed-point problem, particularly in models with abrupt transitions, where optimal choices are nonlinear and evolve over time. Wang et al. (2022) \cite{wang2022dynamic} explores this complexity by providing analytical solutions for a dynamic oligopoly model, achieving tractability by assuming exogenous timing of firm price adjustments, similar to the Calvo framework \citep{calvo1983staggered}. Our approach aligns with the goals of \cite{caplin1991state} and \cite{wang2022dynamic} in pursuing analytical tractability. However, it diverges by examining a setting where firm-level decisions are endogenously determined by prevailing conditions and where idiosyncratic shocks play a significant role in shaping firm behavior.

Our treatment of strategic complementarities builds on the foundational work of \cite{caplin1991state}, where a firm’s profit function depends on both its own markup and the average markup across firms. A key distinction in our approach is the inclusion of idiosyncratic shocks, a feature absent in their framework. While Caplin and Leahy analyzed equilibrium under aggregate nominal shocks modeled as a drift-less Brownian motion, our focus shifts to studying the impulse response to a one-time shock with a non-zero drift. This adaptation allows for greater flexibility in exploring a more generalized economic system \citep{pramanik2016,pramanik2021thesis}.

Our work is closely related to \cite{klenow2016real} and \cite{nakamura2018high}, both of which develop Dynamic Stochastic General Equilibrium (DSGE) models with an input-output framework. These models incorporate the sticky prices of other industries into each industry's costs, highlighting \emph{macro strategic complementarities}. Like these studies, our analysis includes a frictionless labor market, firm-level idiosyncratic shocks, and menu costs associated with price adjustments. Notably, \cite{nakamura2018high} account for random menu costs, a feature also present in our model, while \cite{klenow2016real} incorporate variable demand elasticity at the firm level, referred to as \emph{micro-strategic complementarities}, which we also consider. We establish that, to a second-order approximation, the combined influence of micro and macro complementarities can be summarized by a single parameter. While \cite{klenow2016real} and \cite{nakamura2018high} rely on numerical methods to analyze the effects of monetary shocks on aggregate output, our research provides analytical insights into these dynamics.

Our study shares relevance with the work of \cite{wang2022dynamic}, who examine shock propagation in a sticky-price economy characterized by strategic complementarities. Their analytical solution is based on firms following a time-dependent pricing rule similar to the Calvo framework. Several elements of their model align with ours, such as the role of factors like variable demand elasticity, diminishing returns, and non-zero Frisch elasticity, which are summarized by a single parameter. However, key differences distinguish the two approaches. First, \cite{wang2022dynamic} analyze a dynamic oligopoly without idiosyncratic shocks, whereas our model focuses on oligopolistic competition with idiosyncratic shocks—a feature that improves consistency with observed price change distributions in empirical data. Second, their framework assumes exogenous timing of price adjustments, while ours allows firms to endogenously determine both the timing and magnitude of their price changes. The absence of idiosyncratic shocks and the simplification of exogenous timing in their model facilitate connections to the New Keynesian Phillips curve and enable an exploration of the role of strategic complementarities. 

\cite{bertucci2018optimal} adds to the mean field game (MFG) literature by examining an impulse control problem. However, his study focuses on a more straightforward scenario involving a decision-maker who considers a single adjustment with a fixed adjustment target. Furthermore, his work primarily centers on proving the existence and uniqueness of solutions, using a slightly different concept of what constitutes a solution.

\section{Preliminaries:}
\subsection{Basic Calvo Model:}
This section explores an issue concerning strategic complementarities and pricing, as originally discussed by \cite{calvo1983staggered}. This particular scenario garners frequent attention in studies on sticky prices owing to its practical significance. The model offers a clear-cut framework for elucidating the core elements of the analysis and for scrutinizing crucial outcomes such as the existence, uniqueness, and non-monotonic characteristics of impulse response profiles, which also bear relevance to the state-dependent problem \citep{hua2019}.

Let at time $s$, $Z(s)$ be the consumer price index (CPI), $V_i(s)$ be a consumer preference shock corresponding to  $i^{th}$ variety and the price set by a firm on consumer good of $i^{th}$ variety be $\hat Z(s)$ such that $z(s)\equiv\hat z(s)/V_i(s)$. Define $\tilde x(s):=\frac{ z(s)-\hat Z(s)}{Z(s)}$ and $\tilde X(s):=\frac{Z(s)-\hat Z(s)}{Z(s)}$ as percent deviation from symmetric equilibrium of a firm's own and the aggregate price (CPI), respectively. The economy comprises a range of \emph{atomistic} individual firms, each operating independently. Each firm operates under the assumption of a consistent fluctuation in markup averages, denoted by $\tilde X(s)\in\mathcal X$ for all times $s\in[0,t]$, where $\mathcal X$ is a functional space taking the values from $\mathbb R^n$. The firm has the ability to adjust its pricing only at specific, randomly occurring times denoted by $\{\xi_i\}$, which follows a Poisson process characterized by a parameter $\theta$. These instances of adjustment are referred to as \emph{adjustment opportunities}, and the state of the firm's pricing at these times is termed the \emph{optimal reset value}. Following a price adjustment at time s, the difference in markup, $\tilde x(s)$, jumps according to a Brownian motion without a drift component but with a variance of $\sigma^2$. Additionally, the markup experiences abrupt jumps immediately after a price adjustment at $s=\xi_i$, with each jump in markup denoted by $\mathcal U_i$. Therefore, the markup gap evolves as
\begin{equation}\label{1}
\tilde X(s)=\tilde X(0)+\sigma\left[{\bf \mathcal W}(s)-{\bf \mathcal W}(0)\right]+\sum_{\xi_i\leq s} \mathcal \mathcal U_i,\ \ \text{for all $s\in[0,t]$}, 
\end{equation}
where $\bf \mathcal W$ is a $d$-dimensional standard Brownian motion. Furthermore, in the absence of any markup jump the continuous version of Equation (\ref{1}) becomes,
\[
\tilde X(t)=\tilde X(0)+\int_0^t\sigma\ d\mathcal W(s)+\sum_{\xi_i\leq s} \mathcal \mathcal U_i.
\]
Throughout our analysis we use the following form of the SDE
\begin{equation}\label{1.1}
d\tilde X(s)=\mu[s,m(s),\tilde X(s)]ds+\sigma[s,m(s),\tilde X(s)]d\mathcal W(s)+\sum_{\xi_i\leq s} \mathcal \mathcal U_i,    
\end{equation}
where $\mu:[0,t]\times\mathcal X\times\mathcal M\mapsto \mathbb R^{n\times n}$ and $\sigma:[0,t]\times\mathcal X\times\mathcal M\mapsto \mathbb R^{n\times d}$ are the drift and diffusion components, respectively, and $m(s)\in\mathcal M$ is the complementarity strategy of a firm due to random monetary shock at time $s$, where an element of the  adaptive control space $\mathcal M$ takes value from $\mathbb R^n$. Throughout the analysis we denote $m(s)$ as adaptive control \citep{polansky2021motif}. An unexpected shift in monetary policy, known as a monetary shock, can lead to strategic complementarities among firms when their responses to the shock mutually reinforce each other. Such shocks, which may involve abrupt changes in the money supply or interest rates, affect key economic factors such as aggregate demand and inflation. Strategic complementarities emerge when the optimal reaction of a firm to the shock is positively influenced by the decisions of other firms.

Monetary shocks create strategic complementarities among firms by\\
(i). Pricing decisions: A positive monetary shock (i.e., lower interest rates or higher money supply) increases aggregate demand. If one firm raises its prices in anticipation of higher demand or inflation, others may follow, reinforcing each other's price-setting behavior. On the other hand, a negative monetary shock (e.g., higher interest rates or reduced money supply) lowers aggregate demand. If one firm lowers prices to maintain competitiveness, others may feel pressured to do the same, creating strategic complementarity in price adjustments. This type of environment is useful in sticky price model where firms' pricing decisions depend on expectations of other firms' price-setting behavior.\\
(ii). Production and investment decisions: A monetary expansion reduces borrowing costs and increases expected demand. If a firm increases production or invests in capacity, it signals higher confidence in future demand, encouraging others to follow.\\
(iii). Hiring and wage decisions: Under monetary expansion, if a firm hires more workers or raises wages in response to an expected demand increase, others would implement the same strategy to compete for labor or to prepare for higher future demand, while a firm cutting jobs or wages during a monetary contraction may make it optimal for other firms to reduce costs similarly to remain competitive.\\
(iv). Consumer expectations: A monetary shock influences consumer spending and saving behavior. Firms responding to these changes strategically may amplify each other's actions by prompting firms to increase advertising or launch new products, creating a positive feedback loop. Contrarily, a monetary contraction leads to firms simultaneously scaling back operations or delaying investments.
The strategic complementarities via monetary shocks are based on conditions like, a firms must operate under a setting where their actions influence each other (e.g., oligopolistic markets, input-output linkages), firms must form expectations about the monetary shock and its effects based on other firms' behavior, and price and wage rigidities, coordination problems, or incomplete information make the complementarities more pronounced.

\subsection{Probalistic Foundation:}
Let $t>0$  be a fixed, finite time. Consider $\mathcal W(s)$, a d-dimensional Wiener process observed during a markup dynamics at time \( s \in [0, t] \). This process is defined on a complete probability space \( (\Omega, \mathcal{F}, \mathcal{P}) \), where \( \Omega \) is the sample space, \( \mathcal{F} \) is the \(\sigma\)-algebra, and \( \mathcal{P} \) is a probability measure.

\begin{defn}\label{d0}
     Let \(\left\{\mathcal{F}_t\right\}_{t \in I}\) be a collection of sub-\(\sigma\)-fields of \(\mathcal{F}\), where \(I\) is an ordered index set satisfying the condition \(\mathcal{F}_s \subset \mathcal{F}_t\) for all \(s < t\) with \(s, t \in I\). This collection, \(\left\{\mathcal{F}_t\right\}_{t \in I}\), is referred to as a markup filtration of the process generated due to strategic complementarities.
  \end{defn}

If we simply consider  the markup dynamics $\{\tilde x\}_{t\in I}$ or simply $\tilde x(t)$, then this implies that the choice of markup filtration corresponding to a strategic complementarities due to monetary shock 
\[
\mathcal F_t:=\mathcal F_t^{\tilde x}:=\sigma\left\{\tilde x(s)\bigg|s\leq t;\  s,t\in I\right\},
\] 
which is termed as canonical or natural markup filtration of ${\tilde x}_{t\in I}$. In our case, the Wiener process is associated with the canonical markup filtration,
\[
\mathcal F_t^{\mathcal W}:=\sigma\left\{\mathcal W(s)\bigg|0\leq s\leq t\right\},\ \ t\in[0,\infty).
\]

\begin{defn}\label{d1}
    A set $\{(\tilde X(t),\mathcal F_t\}_{t\in I}$ with filtration $\left\{\mathcal{F}_t\right\}_{t \in I}$ and a family of $\mathbb R^n$-valued markup $\{\tilde X(t)\}_{t\in I}$ with $\tilde X(t)$ being $\mathcal F_t$-measurable  is defined a stochastic process with markup filtration $\left\{\mathcal{F}_t\right\}_{t \in I}$.
\end{defn}

\begin{defn}\label{d2}
  Let $(\Omega,\mathcal F,\mathcal P)$ be a probability space with filtration \(\left\{\mathcal{F}_t\right\}_{t \in I}\). A real valued adapted stochastic markup process $\tilde X(t)$ is a martingale with respect to \(\left\{\mathcal{F}_t\right\}_{t \in I}\) if $\E|\tilde X(t)|<\infty$ for all $t$ and for all $s\leq t$ we have $\E\left\{\tilde X(t)|\mathcal F_s\right\}=\tilde X(s)$.
\end{defn}
 A martingale is a completely random process, characterized by the property that, based on its past behavior, the expected value at any future point in time equals its current value. Note that, $\E\{\tilde X(t)\}=\E\{X(0)\}$ for all $t\in[0,\infty)$.

 \begin{lem}\label{l0}
 (i). The Wiener process $\{\mathcal W(t)\}_{t\in[0,\infty)}$ is an $\mathcal F_t^\mathcal W$ martingale.\\
 (ii). The Process $\left\{\mathcal W^2(t)-t\right\}_{t\in[0,\infty)}$ is an $\mathcal F_t^\mathcal W$ martingale.\\
 (iii). $\left\{\exp\left[\sigma\mathcal W(t)-\frac{1}{2}\sigma^2 t\right]\right\}_{t\in[0,\infty)}$ is an $\mathcal F_t^\mathcal W$ martingale.
 \end{lem}
 \begin{proof}
    See in the Appendix.
\end{proof}

\begin{rmk}\label{r0}
    Lemma \ref{l0} highlights the martingale properties of the Wiener process and its related transformations. It shows that the Wiener process, its squared process adjusted for time, and a specific exponential transformation all satisfy the martingale property under the natural filtration. 
\end{rmk}
\begin{lem}\label{l1}
   Let the markup $\tilde X(t)$ for all $t\in[0,\infty)$ be a super martingale. Then for a constant non-negative monetary shock $\eta$, following inequality holds
   \begin{equation*}
       \eta P\left[\inf_t\tilde X(t)\leq-\eta\right]\leq\sup_t\E\bigg\{\max\left(-\tilde X(t),0\right)\bigg\},
   \end{equation*}
   for each $\eta\geq 0$.
\end{lem}
\begin{proof}
    See in the Appendix.
\end{proof}
\begin{prop}\label{p1}
 Let $\left\{\tilde X(t)\right\}_{t\in[0,\infty)}$  be a martingale. Then for a constant non-negative monetary shock we have
 \[
 \eta P\left[\sup_t |\tilde X(t)|\leq -\eta\right]\leq\sup_t||\tilde X(t)||_1.
 \]
\end{prop}
\begin{proof}
    Jensen's inequality implies, if markup $\tilde X(t)$ is a martingale, then $\tilde B(t)=-|\tilde X(t)|$ is a negative martingale such that $||\tilde B(t)||_1=||\tilde X(t)||_1=\E\bigg\{\max\left(-\tilde B(t),0\right)\bigg\}$. Moreover, 
    \[
    \left[\inf_t\tilde B(t)\leq-\eta\right]\leq\left[\sup_t|\tilde X(t)|\geq\eta\right].
    \]
    The result follows by Lemma \ref{l1}.
\end{proof}

\begin{prop}\label{p2}
  Consider two markups $\tilde X$ and $\hat X$ defined on $(\Omega,\mathcal F,\mathcal P)$ so that for some $d\in(1,\infty)$ assume $\tilde X\in L^d$. For any positive monetary shock if
  \[
  \eta P\left[(\hat X\geq\eta)\right]\leq\int_{(\hat X\geq\eta)}\tilde XdP,
  \]
  then $||\hat X||_d\leq\tilde d\ ||\tilde X||_d,$ such that $1/d+1/\tilde d=1$.
\end{prop}
\begin{proof}
    See in the Appendix.
\end{proof}

\begin{prop}\label{p3}
    For all $t\in[0,\infty)$ let markup $\tilde X(t)$ be right continuous submartingale. Define $\bar X(\omega):=\sup_t\tilde X(t,\omega)$. Then for a given market shock $\eta\in(1,\infty]$, and $\bar X\in L^d$ iff $\sup_t||\tilde X(t)||_d<\infty$. Moreover, if $(\tilde d)^{-1}=1-d^{-1}$, and the market shock is positive and finite, then $||\bar X||_d\leq\tilde d\sup_t||\tilde X(t)||_d$.
\end{prop}
\begin{proof}
    See in the Appendix.
\end{proof}

\begin{cor}\label{c0}
 (Doob's inequality) If $d=\tilde d=2$ is defined in continuous time interval $(0,t)$, and $\tilde X(s)\}_{s\geq 0}$ is a martingale, then
 \[
 \E\left\{\sup_{s\in[0.t]}|\tilde X(s)|^2\right\}\leq 4\E\bigg\{|\tilde X(t)|^2\bigg\}.
 \]
\end{cor}
In this paper we are considering a Cox-Ingersoll-Ross (CIR) version of the SDE expressed in Equation \eqref{1.1}. The markup dynamics is represented as 
\begin{equation}\label{1.2}
  d\tilde X(s) = \bigg\{ \tilde\theta \left[u - \tilde X(s) \right] + m^2(s) \bigg\} ds + \sigma \sqrt{\tilde X(s)} \mathcal W(s)+ \sum_{k:T_k\leq \tilde s}J_{(k)},  
\end{equation}
where $\tilde\theta$ is mean reversion rate constant, u is the long term mean of the markup dynamics, $m(s)$ represent strategic complementarity due to monetary shock, $\sigma$ is homoscedastic variance in $[0,t]$, and $J_{(k)}$ is amount of markup jump right after a price change at time $s^*=T_k$, for all $\tilde s\in[ s^*,t]$.

Similar to \cite{alvarez2023price} we assume that the strategic complementarities due to monetary shock at work up the finite time horizon $t$. Moreover, for $s<t$, by Proposition 1 of \cite{alvarez2023price} the period flow cost is $E\big[\tilde x(s)+\be \tilde X(s)\big]^2$, such that 
\[
E\equiv\frac{1}{2}\bigg[\rho'(1)+\rho(1)\big(\rho(1)-1\big)\bigg]>0,\ \text{and}\ \be=-\frac{\hat Z}{z^*}\frac{\partial z^*}{\partial {\hat Z}},
\]
where $\rho$ is the elasticity of demand with respect to the own price $z$. Proposition 1 from \cite{alvarez2023price} establishes that a firm's dynamic profit maximization is equivalent to minimizing $E\big[\tilde{x}(s) + \beta \tilde{X}(s)\big]^2$.
Our contribution is to introduce an additional adaptive control variable, \(m(s)\), which represents strategic complementarities arising from monetary shocks. Notably, the degree of strategic interaction between a firm's own price and the aggregate price level is determined by the parameter \(\beta\). Static profit maximization occurs when \(\tilde{x} = -\beta\tilde{X}\). If \(\beta < 0\), the firm experiences strategic complementarities, while \(\beta > 0\) indicates strategic substitutability. Consequently, if \(\beta \neq -1\), the sole static equilibrium is achieved at \(\tilde{X}(s) = 0\).

Additionally, in the absence of macroeconomic complementarity, i.e., when \(\partial\tilde{X}/\partial Z(s) = 0\),
\[
\beta = -\frac{\rho'}{\rho(\rho - 1) + \rho'},
\]
where \(\beta < 0\) if \(\rho' > 0\). Economically, this implies that a positive \(\rho'\) reduces demand elasticity as \(Z\) increases, prompting the firm to raise its markup. Thus, when \(\rho' > 0\), a firm's price and the aggregate price are strategic complements. Moreover, if \(\partial\tilde{X}/\partial Z(s) = 0\), the extent of strategic complementarities is limited by the condition \(\beta > -1\). In contrast, if \(\partial\tilde{X}/\partial Z(s) > 0\), the condition \(\beta < -1\) holds.

\begin{as}\label{as0}
  There exist deterministic constants $K_1$ and $K_2$  and another real valued markup  $\kappa$ such that for two different complementarities $m(s)$ and $\underline{m}(s)$, and for all $\omega\in\Omega$ following conditions hold
  \begin{align*}
  \bigg|\bigg\{\tilde\theta \left[u - \tilde X(s,\omega) \right] + m^2(s,\omega)\bigg\}- & \bigg\{\tilde\theta \left[u - \bar X(s,\omega) \right] + \underline{m}^2(s,\omega)\bigg\}\bigg|+\sigma\left| \sqrt{\tilde X(s,\omega)}- \sqrt{\bar X(s,\omega)}\right|\\
  &\hspace{1cm}\leq K_1\left|\tilde X(s,\omega)-\bar X(s,\omega)\right|+K_2\left|{m}(s,\omega)-\underline{m}(s,\omega)\right|,\\
  \bigg|\bigg\{\tilde\theta \left[u - \tilde X(s,\omega) \right] + m^2(s,\omega)\bigg\}-& \sigma \sqrt{\tilde X(s,\omega)}\bigg|\leq\kappa(t,\omega)+K_1\big|\tilde X(s,\omega)\big|,
  \end{align*}
  with $\E\left\{\int_0^t\big|\kappa(s)|^2ds\right\}<\infty$, for all $t\in[0,\infty)$.
\end{as}

\begin{rmk}\label{r1}
 Existence and uniqueness of a strong solution of CIR is ensured by Assumption \ref{as0}. Under the first inequality of Assumption \ref{as0}, an obvious selection of $\kappa$ would be $\kappa(t)=\big|\tilde\theta u+m^2\big|$, once it satisfies $\E\left\{\int_0^t\big|\kappa(s)|^2ds\right\}<\infty$.
\end{rmk}

\begin{prop}\label{p0}
  Under Assumption \ref{as0}, Equation \eqref{1.2} has a unique solution.
\end{prop}
\begin{proof}
    See in the Appendix.
\end{proof}

Moreover, if all conditions of Assumption \ref{as0} are satisfied, the strong solution to the SDE described by Equation \eqref{1.2} exists over the continuous time interval $[0,t]$. Additionally, for any $\mathcal F_0$-measurable random initial value $\tilde x_0$ in $\mathbb{R}^n$ such that $\E\left\{|\tilde x_0|^d\right\}<\infty$ for some $d>1$, there exists a unique strong solution ${\tilde X}^*(s)$ beginning from $\tilde x_0$ at time $0$. Specifically, this means ${\tilde X}^*(0) = \tilde x_0$. The uniqueness is pathwise, implying that if ${\tilde X}^*(s)$ and ${\bar X}^*(s)$ are two strong solutions, then $P\left[{\tilde X}^*(s) = {\bar X}^*(s) \ \forall\ s \in [0,t]\right] = 1$. Furthermore, this solution is square-integrable: for all $t \in [0, \infty)$, there exists a constant $K_t$ such that
\[
\E\left\{\sup_{s\in[0,t]}\big|{\tilde X}^*(s)\big|^d\right\}\leq K_t\left[1+\E\big\{|\tilde x_0|^d\big\}\right].
\]

\begin{prop}\label{p4}
  The SDE represented by the Equation  \eqref{1.2} has a solution 
  \begin{align*}
      \tilde{X}^*(s) &= \exp\bigg\{\exp\left\{-\tilde{\theta} s\right\} \bigg[\exp\left\{\tilde{\theta} s_0\right\} \ln(\tilde{X}(s_0)) + \int_{s_0}^s \exp\left\{\tilde{\theta} \xi\right\} \bigg\{\frac{\tilde{\theta} u}{\tilde{X}(\xi)} + \frac{m^2(\xi)}{\tilde{X}(\xi)}\bigg\} d\xi\\
      &\hspace{1cm}+ \int_{s_0}^s \exp\left\{\tilde{\theta} \xi\right\} \frac{\sigma}{\sqrt{\tilde{X}(\xi)}} \, d\mathcal{W}(\xi) + \int_{s_0}^s \exp\left\{\tilde{\theta} \xi\right\} \frac{1}{\tilde{X}(\xi)} \sum_{k:T_k \leq \xi} J_{(k)} d\xi\bigg]\bigg\},
\end{align*}
where the integrating factor is $\exp\left\{\tilde{\theta} s\right\} /\tilde{X}(s)$.
\end{prop}
\begin{proof}
    See in the Appendix.
\end{proof}

In this case we consider the objective function of the firm as
\begin{equation}\label{obj}
 \mathcal J(\tilde x,m)=\E\left\{\int_0^t\exp(-\rho s)\bigg[\xi\left[\tilde X(s)+m(s)+\phi\tilde X(s)\right]^2+c(m(s))\bigg]ds\bigg |\mathcal F_0\right\},  
\end{equation}
where $\rho$ is the discount factor, $c(m(s))$ is the cost associated  with the complementarity strategy of a firm due to random monetary shock, and $\xi$ and $\phi$ are constants. For $(s,\tilde x)\in[0,t]\times\mathbb R^n$, define $\mathcal M(s,\tilde x)$ the subset of $m$'s in $\mathcal M$ such that
\[
\E\left\{\bigg|\xi\left[\tilde X(s)+m(s)+\phi\tilde X(s)\right]^2+c(m(s))\bigg|ds\right\}<\infty,
\]
and assume that $\mathcal M(s,\tilde x)$ is non-empty. The objective is to over $m$ the gain function $\mathcal J$ and the associated value function is defined as 
\begin{equation}\label{obj0}
v(s,\tilde x)=\sup_{m\in\mathcal M(s,\tilde x)}  \mathcal J(\tilde x,m).  
\end{equation}

\subsection{Jump Diffusion:}
In this section we discuss about the jump diffusion in our context. The incorporation of jump diffusion in Equation \eqref{1.2} for the state variable \( \tilde{X}(s) \), which represents the percent deviation of aggregate prices (CPI) from the symmetric equilibrium, plays a fundamental role in modeling the complex dynamics of an economy influenced by stochastic shocks and price adjustments. In this setup, the state variable \( \tilde{X}(s) \) captures the interplay between individual firm pricing decisions and aggregate economic outcomes, bridging the micro and macroeconomic perspectives. The economy is modeled as comprising a continuum of atomistic firms, each independently adjusting their pricing strategies \( \hat{z}(s) \) based on individual preference shocks \( V_i(s) \), while simultaneously responding to aggregate market conditions. The deviation \( \tilde{X}(s) \) encapsulates the dynamic behavior of the overall pricing structure, influenced not only by continuous, smooth adjustments driven by mean-reverting forces and strategic complementarities but also by abrupt, discrete jumps resulting from external shocks. 

The jump diffusion component is particularly critical as it captures the non-continuous changes in markup averages, which arise from events such as monetary policy shifts, fiscal interventions, demand surges, supply disruptions, or other exogenous shocks to the economic environment. These jumps introduce a layer of realism to the model, as they reflect the sudden, large-scale adjustments that firms must make in response to unforeseen changes in market conditions. By modeling these jumps, the SDE can account for the often unpredictable and nonlinear nature of economic adjustments, offering a more comprehensive understanding of how aggregate price levels evolve over time. Furthermore, the control variable \( m(s) \), representing a firm’s complementarity strategy in response to random monetary shocks, adds an adaptive dimension to the model, illustrating how firms strategically align their pricing behavior with prevailing economic conditions \citep{pramanik2024dependence,pramanik2024measuring}. The combination of the diffusive component, which models gradual adjustments through mean reversion and volatility in pricing, and the jump component, which captures the stochastic, discrete shifts in markups, provides a robust framework for analyzing the dynamics of aggregate price deviation. 

This approach highlights the dual nature of price-setting behavior: a steady state-seeking tendency interspersed with large-scale, disruptive events. Importantly, the jump diffusion framework allows for a nuanced exploration of how individual firm-level decisions, when aggregated across a diverse economy, contribute to broader patterns of price instability, volatility, and recovery. Such a model is invaluable for policymakers, as it offers insights into the transmission mechanisms of monetary policy, the impact of shocks on price stability, and the effectiveness of interventions aimed at restoring equilibrium in the face of economic disruptions. By capturing both the continuous and discontinuous elements of economic fluctuations, the jump diffusion extension of the CIR SDE provides a powerful tool for studying the dynamic interplay between microeconomic decision-making and macroeconomic outcomes, offering a rich perspective on the forces shaping aggregate price movements in a stochastic and shock-prone economic environment. We will anlyse some properties of jump diffusion.

\begin{lem}\label{j0}
  Let the jump diffusion of the SDE expressed by Equation \eqref{1.2} follows a Poisson process with intensity \( \nu \), jump sizes \( J_{(k)} \) with mean \( \mathbb{E}[J_{(k)}] = \gamma \), and a finite time horizon \( [0,t] \). The expected contribution of the jumps to the aggregate price dynamics is given by:
\[
\mathbb{E} \left[ \sum_{k:T_k \leq t} J_{(k)} \right] = \mathbb{E}[N(t)] \cdot \gamma,
\]
where \( N(t) \) is the number of jumps by time \( t \), and \( \mathbb{E}[N(t)] = \nu t \).  
\end{lem}
\begin{proof}
    See in the Appendix.
\end{proof}

\begin{rmk}\label{r2}
 Lemma \ref{j0} formalizes the expected contribution of these jumps, demonstrating that it is determined by the expected number of jumps, \( \mathbb{E}[N(t)] \), and the average jump size, \( \gamma = \mathbb{E}[J_{(k)}] \). This result highlights the additive role of jumps in shaping aggregate price dynamics, showing how their expected magnitude and frequency influence the trajectory of \( \tilde{X}(s) \).
   \end{rmk}

\begin{lem}\label{j1}
   Under the same assumptions as Lemma \ref{j0}, the variance of the total jump contribution to \( \tilde{X}(t) \) is given by:
\[
\text{Var} \left( \sum_{k:T_k \leq t} J_{(k)} \right) = \nu t (\gamma^2 + \sigma_J^2),
\]
where \( \gamma = \mathbb{E}[J_{(k)}] \) and \( \sigma_J^2 = \text{Var}(J_{(k)}) \) is the variance of the jump size.
\ 
\end{lem}

\begin{proof}
    See in the Appendix.
\end{proof}

\begin{rmk}\label{r3}
   Moreover, Lemma \ref{j1} provides further insight by quantifying the variance of the jump contribution, which is determined by both the second moment of the jump size \( (\gamma^2 + \sigma_J^2) \) and the intensity \( \nu \) of the jump process. This variance is crucial for understanding the impact of jumps on economic volatility, as it captures how stochastic fluctuations in jump size and frequency propagate through the system. Together, these Lemmas provide a rigorous foundation for analyzing the dual role of jumps: while their expected contribution shifts the aggregate price level, their variance amplifies the uncertainty in the pricing environment.
 \end{rmk}

 \begin{prop}\label{p5}
For SDE in Equation \eqref{1.2}, the joint distribution of \( \tilde{X}(t) \) and the cumulative jump process \( \sum_{k: T_k \leq t} J(k) \) at time \( t \) converges to a bivariate Gaussian distribution in the limit as \( t \to \infty \), such that  \( J(k) \) are i.i.d. with mean \( \E[J(k)] = \gamma \) and variance \( \text{Var}(J(k)) = \sigma_J^2 \), and the intensity \( \nu \) of the Poisson jump process satisfies \( \nu t \to \infty \) as \( t \to \infty \). In particular, the mean and covariance structure of the joint distribution of \( \left( \tilde{X}(t), \sum_{k: T_k \leq t} J(k) \right) \) are given by:

\[
\E[\tilde{X}(t)] \to u, \quad \E\left[\sum_{k: T_k \leq t} J(k)\right] \to \nu t \gamma,
\]
and
\[
\text{Cov}(\tilde{X}(t), \sum_{k: T_k \leq t} J(k)) \to 0, \quad \text{Var}(\tilde{X}(t)) \to \frac{\sigma^2}{2 \tilde{\theta}}, \quad \text{Var}\left(\sum_{k: T_k \leq t} J(k)\right) \to \nu t \cdot \left(\gamma^2 + \sigma_J^2\right).
\]
 \end{prop}

 \begin{proof}
    See in the Appendix.
\end{proof}

 The implication of the Proposition \ref{p5} is that, as time progresses, the dynamics of the state variable and cumulative jump process stabilize and converge to a bivariate Gaussian distribution. This indicates that in the long run, the system's behavior becomes predictable, with the state variable approaching its long-term mean and the cumulative jumps growing linearly with time. The variance of both components reaches stable values, and they become independent, reflecting the decoupling of continuous price dynamics from the discrete shocks. This convergence to Gaussianity simplifies forecasting and analysis, allowing policymakers and economists to better understand and manage long-term economic outcomes, especially in the presence of random shocks.

 \begin{prop}\label{p6}
({\bf Lyapunov stability})  The process \( \tilde{X}(s) \) is asymptotically stable at \( u \) as \( s \to \infty \), such that\\
(i). The jump sizes \( J(k) \) are i.i.d. with mean \( \E[J(k)] = \gamma \) and variance \( \text{Var}(J(k)) = \sigma_J^2 \).\\
(ii). The intensity \( \nu \) of the Poisson jump process satisfies \( \nu t \to \infty \) as \( s \to \infty \).\\
(iii). The control function \( m^2(s) \) is bounded over time.
\end{prop}
\begin{proof}
    See in the Appendix.
\end{proof}

\begin{rmk}\label{r4}
   Proposition \ref{p6} implies that the process \( \tilde{X}(s) \) is Lyapunov stable around the equilibrium point \( u \). That is, any small deviation of \( \tilde{X}(s) \) from \( u \) will decay over time, and the process will remain close to \( u \) for all \( s \). The mean reversion driven by the term \( -\tilde{\theta} (\tilde{X}(s) - u) \) ensures stability, and the fluctuations due to jump diffusion and random monetary shocks do not lead to instability under the condition that the jump intensity \( \nu \) is sufficiently large and the shocks remain bounded.

\ 
\end{rmk}

\subsection{Meanfield Games:}
The mean-field game (MFG) approach provides a robust and versatile framework for analyzing strategic decision-making among a large population of agents, making it especially valuable in economic models where aggregate outcomes emerge from individual behaviors. In the context of studying firm-level pricing decisions under monetary shocks, where the state variable represents the percentage deviation from the symmetric equilibrium of a firm’s aggregate price (CPI), and the control variable captures strategic complementarities, the MFG framework is uniquely suited to unraveling the complexities of such interactions\citep{pramanik2023cont,pramanik2024estimation}. The deviation from equilibrium encapsulates the tension between firms' pricing adjustments and the overall market dynamics, reflecting how monetary shocks disrupt the expected balance of supply and demand. Firms, in this scenario, face the dual challenge of responding optimally to aggregate shocks while accounting for the actions of competing firms. The control variable, representing strategic complementarities, signifies how a firm's pricing behavior is influenced by expectations of others' pricing adjustments, a phenomenon that becomes especially pronounced during periods of monetary instability \citep{pramanik2024bayes}.

The importance of this approach lies in its ability to model these interdependencies at both the micro and macro levels. The MFG framework enables firms’ individual pricing decisions to be analyzed within the broader context of collective market behavior, providing critical insights into how monetary policy shocks ripple through the economy. Strategic complementarities play a pivotal role here, as they amplify the coordination problem among firms. For instance, when a monetary shock induces a firm to raise its prices, it can lead to a cascading effect where other firms are incentivized to follow suit to maintain their competitive positioning \citep{pramanik2023cmbp}. Conversely, under certain conditions, the same complementarities may mitigate the shock’s impact by fostering stabilization, as firms collectively align their strategies toward the equilibrium state. This delicate interplay between individual and collective decision-making is central to understanding the propagation of shocks and their aggregate consequences.

Moreover, the MFG framework offers a computationally feasible way to derive equilibria in such settings, particularly when dealing with a large number of agents. By focusing on the limit behavior of the system as the number of agents approaches infinity, MFG models reduce the complexity of analyzing individual interactions while capturing the aggregate dynamics accurately \citep{pramanik2023path}. This is particularly important in the context of monetary shocks, where traditional models may struggle to account for the nonlinear and feedback-driven nature of strategic interactions among firms. The equilibrium solutions provided by MFG models shed light on the steady-state behavior of prices, helping to explain how firms converge back to equilibrium—or fail to do so—following a shock. Additionally, this approach facilitates the exploration of welfare implications by quantifying the costs of deviations from equilibrium and the effectiveness of policy interventions aimed at mitigating such deviations \citep{pramanik2021consensus}.

In practice, the MFG approach also serves as a valuable tool for policymakers seeking to design effective monetary policies. By simulating different shock scenarios and observing the resulting price dynamics, policymakers can identify strategies that minimize inflation volatility and ensure smoother transitions back to equilibrium \citep{pramanik2021}. Furthermore, the framework can be extended to incorporate heterogeneity among firms, such as differences in production costs or market power, providing a more nuanced understanding of how shocks affect various segments of the economy differently. This capability is critical for designing targeted interventions that address specific vulnerabilities while maintaining overall market stability \citep{pramanik2023optimization001}. The MFG approach’s ability to bridge the micro-level behaviors of individual firms with macro-level economic outcomes underscores its importance in modern economic analysis, making it an indispensable tool for understanding and managing the intricate dynamics of aggregate price adjustments in response to monetary shocks.

While the MFG approach provides a powerful framework for analyzing large-scale strategic interactions, it is not without its limitations, particularly in the scenario where the state variable represents the percentage deviation from the symmetric equilibrium of a firm’s aggregate price (CPI) and the control variable reflects strategic complementarities due to monetary shocks. One significant limitation of the MFG approach is its reliance on the assumption of a continuum of agents, which simplifies the complex interactions among firms but often overlooks critical heterogeneities in the market. Real-world firms vary widely in size, market power, and responsiveness to monetary shocks, yet MFG models typically aggregate these variations into a homogeneous population. This can lead to oversimplified conclusions about the aggregate dynamics and fail to capture important disparities in behavior and outcomes. For instance, smaller firms may react differently to monetary shocks compared to larger, more established firms with greater market influence, yet such nuances are often lost in mean-field approximations. Additionally, the assumption that individual firms are negligible in their impact on the overall system may not hold in markets dominated by a few key players, where the actions of a single firm can significantly influence aggregate outcomes. Moreover, this construction can be used in cancer research \citep{dasgupta2023frequent,hertweck2023clinicopathological,kakkat2023cardiovascular,khan2023myb,khan2024mp60,vikramdeo2024abstract,vikramdeo2023profiling}

Another limitation lies in the treatment of information and expectations within the MFG framework. These models often assume that firms have perfect or near-perfect knowledge of the system's dynamics and the strategies of other firms \citep{pramanik2022lock}. However, in practice, firms operate under conditions of uncertainty and incomplete information, particularly during monetary shocks when the economic environment is highly volatile. Strategic complementarities further complicate this issue, as firms must form expectations not only about the shock itself but also about how other firms will adjust their prices in response. These interdependencies can lead to coordination failures or misaligned expectations, phenomena that are difficult to model accurately within the standard MFG framework \citep{pramanik2021optimala}. Furthermore, the assumption of a smooth and well-defined equilibrium may not hold in scenarios with strong strategic complementarities, where multiple equilibria or non-equilibrium dynamics, such as oscillations or chaotic behavior, can emerge. Such complexities are challenging to capture within the MFG framework, which typically focuses on identifying a single, steady-state solution \citep{pramanik2021scoring}.

The computational tractability of MFG models, often seen as a strength, can also become a limitation in this context. While MFG reduces the complexity of analyzing a large number of agents by focusing on the limiting behavior as the population size tends to infinity, this simplification may fail to capture critical dynamics in finite systems. For example, in real-world markets, the number of firms is finite, and discrete interactions can produce outcomes that differ significantly from the predictions of the mean-field limit \citep{pramanik2020optimization,pramanik2023semicooperation}. These discrepancies are particularly pronounced in situations where monetary shocks induce large deviations from equilibrium, pushing the system into regimes where the mean-field approximation no longer holds \citep{pramanik2024estimation,vikramdeo2024mitochondrial}. Additionally, the mathematical complexity of MFG models can pose challenges when incorporating realistic features such as stochastic volatility, nonlinear demand functions, or firm-specific shocks, which are crucial for accurately modeling aggregate price dynamics under monetary policy \citep{pramanik2024motivation}.

Lastly, the policy implications derived from MFG models may be limited by their reliance on idealized assumptions. For example, while MFG models provide insights into how firms collectively adjust prices in response to monetary shocks, they often abstract away from institutional and regulatory factors that influence pricing behavior in practice \citep{pramanik2020motivation}. Factors such as price rigidity, menu costs, or the role of government interventions in stabilizing markets are typically not included in standard MFG formulations, yet they play a crucial role in shaping the real-world outcomes of monetary policy \citep{pramanik2024parametric}. These omissions can lead to policy recommendations that are theoretically sound but impractical or ineffective in real-world scenarios. Consequently, while the MFG approach is a valuable tool for understanding aggregate price dynamics and strategic complementarities, its limitations highlight the need for complementary modeling approaches and empirical validation to ensure that its insights are both accurate and actionable \citep{pramanik2022stochastic}.

\section{Computation of optimal strategic complementarities:}

In this section, we will construct a stochastic Lagrangian based on the system consisting of Equations \eqref{obj} and \eqref{1.1}. Before constructing a Feynman-type path integral control \citep{pramanik2020optimization}, we determine the value of long-term mean $u$ empirically. This $u$ can be estimated by fitting the drift term which is $\tilde{\theta} [u - \tilde X(s)] + m^2(s),$ where $\tilde{\theta}$ is the mean-reversion speed. Moreover, if $ m^2(s)\ra 0$, then 
$\mu\ra\tilde{\theta} [u - \tilde X(s)]$.

For discrete case, the average drift between successive time points can be used to estimate $u$. For each time-point $i$, we have that
\[
\mu \approx \frac{\tilde X(s_{i+1}) - \tilde X(s_i)}{s_{i+1} - s_i},
\]
which implies
\[
\frac{\tilde X(s_{i+1}) - \tilde X(s_i)}{s_{i+1} - s_i} = \tilde{\theta} (u - \tilde X(s_i)).
\]
Rearranging the terms yield
\[
u \approx \tilde X(s_i) + \frac{\tilde X(s_{i+1}) - \tilde X(s_i)}{\tilde\theta\left(s_{i+1} - s_i\right)}.
\]
Averaging the estimates from all consecutive time points results
\[
\hat{u} = \frac{1}{N-1} \sum_{i=1}^{N-1} \left[ \tilde X(s_i) + \frac{\tilde X(s_{i+1}) - \tilde X(s_i)}{\tilde\theta\left(s_{i+1} - s_i\right)} \right],
\]
where $N$ is the total number of observations. Now before formulating stochastic Lagrangian, define the objective function in Equation \eqref{obj} as
\begin{equation}\label{obj1}
    \mathcal J(\tilde x,m)=\E\left\{\int_0^t\exp(-\rho s)\Theta[s,m(s),\tilde X(s)]ds\bigg |\mathcal F_0\right\},
\end{equation}
such that $\Theta[s,m(s),\tilde X(s)]:=\xi\left[\tilde X(s)+m(s)+\phi\tilde X(s)\right]^2+c(m(s))$. By \cite{ewald2024adaptation} the stochastic Lagrangian becomes
\begin{multline}\label{obj2}
\tilde{\mathcal{L}}\left(s,\tilde{x},\lambda,m\right)=\E\biggr\{\int_0^t \bigg\{\exp(-\rho s)\Theta\left[s,\tilde X(s),m(s)\right]ds\biggr\}ds\\
+\int_0^t\left[\tilde x(s)-\tilde x_0-\int_0^s\bigg[\mu[\nu,m(\nu),\tilde X(\nu)]d\nu-\sigma[\nu,m(\nu),\tilde X(\nu)]dW_\nu\bigg]\right] d\lambda(s)\biggr\},
\end{multline}
where $\lambda(s)$ is the time-dependent Lagrangian~multiplier.
\begin{theorem}\label{t0}
    For an atomistic firm, if $\{ \tilde X(s), s\in[0,t]\}$ is a markup dynamics then, the optimal strategic complementarity due to monetary shock as a feedback Nash equilibrium $\big\{m^{*}(s,\tilde x)\in\mathcal M\big\}$ would be the solution of the following~equation
\begin{align}\label{obj3}
	\mbox{$\frac{\partial}{\partial m}$}\ell(s,m,\tilde X) \left[\mbox{$\frac{\partial^2}{\partial (\tilde X)^2}$}\ell(s,m, \tilde X)\right]^2=2\mbox{$\frac{\partial}{\partial \tilde X}$}\ell(s,m,\tilde X) \mbox{$\frac{\partial^2}{\partial\tilde X\partial m}$}\ell(s, m,\tilde X),
	\end{align}
	where for an It\^o process $\tilde h(s,\tilde X)\in \mathcal [0,t]\times\mathbb R$
\begin{align}\label{obj4}
	\ell(s, m,\tilde X)&=\exp(-\rho s)\Theta[s,m(s),\tilde X(s)]+\tilde h(s,\tilde X)d\lambda(s)\notag\\
    &+\left[\mbox{$\frac{\partial \tilde h(s,\tilde X)}{\partial s}$}d\lambda(s)+\mbox{$\frac{d\lambda(s)}{d s}$}\tilde h(s,\tilde X)\right]\notag\\
	&\hspace{.25cm}+\mbox{$\frac{\partial \tilde h(s,\tilde X)}{\partial \tilde X}$}\mu\left[s,m,\tilde X\right]d\lambda(s)+\mbox{$\frac{1}{2}$}\left[\sigma\left[s,m,\tilde X\right]\right]^2\mbox{$\frac{\partial^2 \tilde h(s,\tilde X)}{\partial (\tilde X)^2}$}d\lambda(s).
	\end{align}
\end{theorem}
\begin{proof}
    See in the Appendix.
\end{proof}
To demonstrate the preceding Theorem, we present a detailed example to identify an optimal strategic complementarity due to monetary shock  under this environment. Consider a firm has to maximize the expected markup.

 Given our objective function and the SDE below
\begin{equation}\label{obj}
 \mathcal J(\tilde x,m)=\E\left\{\int_0^t\exp(-\rho s)\bigg[\xi\left[\tilde X(s)+m(s)+\phi\tilde X(s)\right]^2+c(m(s))\bigg]ds\bigg |\mathcal F_0\right\},  
\end{equation}
\begin{equation}\label{1.2}
  d\tilde X(s) = \bigg\{ \tilde\theta \left[u - \tilde X(s) \right] + m^2(s) \bigg\} ds + \sigma \sqrt{\tilde X(s)} \mathcal W(s)+ \sum_{k:T_k\leq \tilde s}J_{(k)},  
\end{equation}

\begin{align}
l(s, m, \tilde{X}) &= \exp(-\rho s)\left\{ \xi \left[ \tilde{X}(s) + m(s) + \phi \tilde{X}(s) \right]^2 + c(m(s))\right\}
 + \exp\left( \frac{\tilde{\theta}s}{\tilde{X}(s)} \right) \, d\lambda(s) \notag\\
&\hspace{.5cm} + \frac{\exp\left\{ \tilde{\theta} s \right\} \left( \tilde{\theta} \tilde{X}(s) - \tilde{X}'(s) \right)}{\tilde{X}^2(s)}d\lambda(s) + \frac{d\lambda(s)}{ds} \exp\left( \frac{\tilde{\theta}s}{\tilde{X}(s)} \right) \notag\\
&\hspace{1cm} - \exp\left( \frac{\tilde{\theta}s}{\tilde{X}^2(s)} \left[ \tilde{\theta}(u - \tilde{X}) + m^2(s) \right] \right)  d\lambda(s)\notag \\
&\hspace{1.5cm} + \frac{\sigma}{2\tilde{X}(s)}   \exp(\tilde{\theta} s) \left[ \tilde{\theta}^2 \tilde{X}(s) - \tilde{X}''(s) - 2 \tilde{\theta} \tilde{X}'(s) + \frac{2 \tilde{X}'^2(s)}{\tilde{X}(s)} \right]d\lambda(s),
\end{align}
where 
$\tilde{X}'=d \tilde{X}(s)/ds$ and
$\tilde{X}'' = d^2 \tilde{X}(s)/ds^2$.
Now,
\begin{align*}
\frac{\partial}{\partial m} l(s, m, \tilde{X}) &= \exp(-\rho s) \left\{ 2 \xi \left[ \tilde{X}(s) + m(s) + \phi \tilde{X}(s) \right] + \frac{\partial c(m(s))}{\partial m} \right\},\\
\frac{\partial^2}{\partial \tilde{ X}  \partial m}  l(s,m, \tilde{X}) &= 2 \exp(-\rho s) \xi (1 + \phi),\\
\frac{\partial}{\partial\tilde{X}} l(s, m, \tilde{X}) 
&=2 \xi\exp(-\rho s)   (1 + \phi) \left[ (1 + \phi) \tilde{X}(s) + m(s) \right]-\frac{\tilde{\theta}s}{\tilde{X}^2(s)} \exp\left( \frac{\tilde{\theta}s}{\tilde{X}(s)} \right)   \, d\lambda(s) \\
&\hspace{.5cm} - \exp\left( \tilde{\theta} s \right) \frac{\tilde{\theta} \tilde{X}^2(s) - 2 \tilde{X}(s) \tilde{X}'(s)}{\tilde{X}^4(s)} \, d\lambda(s) \\
&\hspace{1cm} - \frac{\tilde{\theta}s}{\tilde{X}^2(s)} \exp\left( \frac{\tilde{\theta}s}{\tilde{X}(s)} \right)  \frac{d\lambda(s)}{ds} - \exp\left( \frac{\tilde{\theta}s}{\tilde{X}^2(s)} \left[ \tilde{\theta}(u - \tilde{X}) + m^2(s) \right] \right)\\
&\hspace{1.5cm} \times \left[ -\frac{2 \tilde{\theta}s}{\tilde{X}^3(s)} \left[ \tilde{\theta}(u - \tilde{X}) + m^2(s) \right] - \frac{\tilde{\theta}(s)}{\tilde{X}^2(s)} \tilde{\theta} \right]\, d\lambda(s) \\
&\hspace{2cm} + \exp\left(\tilde{\theta} s\right)  \bigg\{ -\frac{\sigma}{2\tilde{X}^2(s)} \left[ \tilde{\theta}^2 \tilde{X}(s) - \tilde{X}''(s) - 2 \tilde{\theta} \tilde{X}'(s) + \frac{2 \tilde{X}'^2(s)}{\tilde{X}(s)} \right]\\
&\hspace{2.5cm}+ \frac{\sigma}{2\tilde{X}(s)} \left[ \tilde{\theta}^2 - \frac{2 \tilde{X}'^2(s)}{\tilde{X}^2(s)} \right] \bigg\}  d\lambda(s).
\end{align*}

\begin{equation}
\begin{aligned}
\frac{\partial^2}{\partial \tilde{X}^{2}} l(s,m, \tilde{X}) &=2 \xi (1 + \phi)^2 \exp(-\rho s)   + \exp\left( \frac{\tilde{\theta}s}{\tilde{X}(s)} \right)  \frac{\tilde{\theta}s}{\tilde{X}^3(s)} \, d\lambda(s) \\
&\hspace{.5cm} +\frac{\exp\left( \tilde{\theta} s \right)}{\tilde{X}^8(s)}\bigg\{\left[-2 \tilde{\theta} \tilde{X}(s) + 2 \tilde{X}'(s)\right] \tilde{X}^4(s) + \left[\tilde{\theta} \tilde{X}^2(s) - 2 \tilde{X}(s) \tilde{X}'(s)\right]  4 \tilde{X}^3(s)\bigg\} d\lambda(s) \\
&\hspace{1cm} + \frac{\tilde{\theta}s}{\tilde{X}^3(s)} \exp\left( \frac{\tilde{\theta}s}{\tilde{X}(s)} \right)  \frac{d\lambda(s)}{ds} +\exp\left\{ \frac{\tilde{\theta}s}{\tilde{X}^2(s)} \left[ \tilde{\theta}(u - \tilde{X}) + m^2(s) \right] \right\}\\
&\hspace{1.5cm}\times \left\{ \frac{2 \tilde{\theta}s}{\tilde{X}^3(s)} \left[ \tilde{\theta}(u - \tilde{X}) + m^2(s) \right] + \frac{\tilde{\theta}s}{\tilde{X}^2(s)}  \tilde{\theta} \right\} d\lambda(s) \\
&\hspace{2cm} + \exp\left(\tilde{\theta} s\right)  \biggr\{ \frac{\sigma}{\tilde{X}^3(s)} \left[ \tilde{\theta}^2 \tilde{X}(s) - \tilde{X}''(s) - 2 \tilde{\theta} \tilde{X}'(s) + \frac{2 \tilde{X}'^2(s)}{\tilde{X}(s)} \right]\\
&\hspace{2.5cm}- \frac{\sigma}{2\tilde{X}^2(s)} \left[ \tilde{\theta}^2 - \frac{2 \tilde{X}^{'2}(s)}{\tilde{X}^2(s)} \right] \biggr\} d\lambda(s).
\end{aligned}
\end{equation}
Therefore, using the above values Equation \eqref{obj4} yields,
\begin{align}\label{l.0}
& \exp(-\rho s) \left\{ 2 \xi \left[ \tilde{X}(s) + m(s) + \phi \tilde{X}(s) \right] + \frac{\partial c(m(s))}{\partial m} \right\} \notag \\
&\hspace{.5cm} \times \biggl[ \left( \exp\left\{ \frac{\tilde{\theta}s}{\tilde{X}^2(s)} \left[ \tilde{\theta}(u - \tilde{X}) + m^2(s) \right] \right\} \right) \notag \\
&\hspace{1cm} \times \left\{ \frac{2 \tilde{\theta}s}{\tilde{X}^3(s)} \left[ \tilde{\theta}(u - \tilde{X}) + m^2(s) \right] + \frac{\tilde{\theta}s}{\tilde{X}^2(s)} \tilde{\theta} \right\} d\lambda(s) + D_2 \biggr]^{2} \notag \\
&= 2\biggl\{ 4 \xi^2 \exp(-2\rho s) (1 + \phi)^2 \bigg[ (1 + \phi) \tilde{X}(s) +  m(s)\bigg]    - \exp\left( \frac{\tilde{\theta}s}{\tilde{X}^2(s)} \left[ \tilde{\theta}(u - \tilde{X}) + m^2(s) \right] \right)\notag \\
&\hspace{1.5cm} \times \left[ -\frac{2 \tilde{\theta}s}{\tilde{X}^3(s)} \left[ \tilde{\theta}(u - \tilde{X}) + m^2(s) \right] - \frac{\tilde{\theta}^2s}{\tilde{X}^2(s)} \right] d\lambda(s) + D_1  2\exp(-\rho s)  \xi  (1 + \phi) \biggr\}.
\end{align}
According to Equation \eqref{l.0}
\begin{align*}
D_1 &= -\frac{\tilde{\theta}s}{\tilde{X}^2(s)} \exp\left( \frac{\tilde{\theta}s}{\tilde{X}(s)} \right)  d\lambda(s)- \exp\left( \tilde{\theta} s \right) \frac{\tilde{\theta} \tilde{X}^2(s) - 2 \tilde{X}(s) \tilde{X}'(s)}{\tilde{X}^4(s)} d\lambda(s) \\
&\hspace{.25cm}- \frac{\tilde{\theta}s}{\tilde{X}^2(s)} \exp\left( \frac{\tilde{\theta}s}{\tilde{X}(s)} \right) \frac{d\lambda(s)}{ds} + \exp\left(\tilde{\theta} s\right) \bigg\{ -\frac{\sigma}{2\tilde{X}^2(s)} \biggr[ \tilde{\theta}^2 \tilde{X}(s) - \tilde{X}''(s)\\
&\hspace{.5cm}- 2 \tilde{\theta} \tilde{X}'(s) + \frac{2 \tilde{X}'^2(s)}{\tilde{X}(s)} \biggr] + \frac{\sigma}{2\tilde{X}(s)} \left[ \tilde{\theta}^2 - \frac{2 \tilde{X}'^2(s)}{\tilde{X}^2(s)} \right] \bigg\} d\lambda(s),
\end{align*}
and
\begin{align*}
D_2 &= 2 \xi (1 + \phi)^2 \exp(-\rho s) + \exp\left( \frac{\tilde{\theta}s}{\tilde{X}(s)} \right)  \frac{\tilde{\theta}s}{\tilde{X}^3(s)} \, d\lambda(s) \\
& \hspace{.25cm}+ \frac{\exp\left( \tilde{\theta} s \right)}{\tilde{X}^8(s)} \bigg\{ \left[-2 \tilde{\theta} \tilde{X}(s) + 2 \tilde{X}'(s)\right] \tilde{X}^4(s) \\
&\hspace{.5cm} + \left[\tilde{\theta} \tilde{X}^2(s) - 2 \tilde{X}(s) \tilde{X}'(s)\right] 4 \tilde{X}^3(s) \bigg\} \, d\lambda(s)  + \frac{\tilde{\theta}s}{\tilde{X}^3(s)} \exp\left( \frac{\tilde{\theta}s}{\tilde{X}(s)} \right) \frac{d\lambda(s)}{ds} \\
&\hspace{1cm}+ \exp\left(\tilde{\theta} s\right) \biggr\{ \frac{\sigma}{\tilde{X}^3(s)} \left[ \tilde{\theta}^2 \tilde{X}(s) - \tilde{X}''(s) - 2 \tilde{\theta} \tilde{X}'(s) + \frac{2 \tilde{X}'^2(s)}{\tilde{X}(s)} \right] \\
&\hspace{1.5cm}- \frac{\sigma}{2\tilde{X}^2(s)} \left[ \tilde{\theta}^2 - \frac{2 \tilde{X}^{'2}(s)}{\tilde{X}^2(s)} \right] \biggr\} \, d\lambda(s).
\end{align*}
Given the complexity, the Equation \eqref{l.0} could lead to a high-degree polynomial in $m(s)$. Determining an explicit solution for $m(s)$ might necessitate additional simplifications or assumptions to lower the polynomial's degree. We assume that the influence of $d\lambda(s)$  is negligible [i.e., \( d\lambda(s) \ra 0 \)], this would effectively remove all terms involving $d\lambda(s)$ from the equation. Hence, 

\begin{multline*}
\exp(-\rho s) \left\{ 2 \xi \left[ \tilde{X}(s) + m(s) + \phi \tilde{X}(s) \right] + \frac{d c(m(s))}{d m} \right\} D_2^2 \\
= 2 \left( 2 \xi \exp(-\rho s) (1 + \phi) \left[ (1 + \phi) \tilde{X}(s) + m(s) \right] + D_1 \right) 2 \exp(-\rho s) \, \xi \, (1 + \phi).
\end{multline*}
Since $c(m(s)) =  c_0 m^2(s)/2$, then $ d c(m(s))/d m = c_0 m(s)$,
 and furthermore, as $ D_1  , D_2$ are not dependent on $m(s)$ ,
 \begin{align*}
&\exp(-\rho s) \bigg\{ 2 \xi \left[ (1 + \phi) \tilde{X}(s) + m(s) \right] + c_0 m(s) \bigg\} D_2^2 \\
&= 4 \exp(-\rho s) \xi (1 + \phi)\bigg[2 \xi \exp(-\rho s) (1 + \phi) \left[ (1 + \phi) \tilde{X}(s) + m(s) \right] + D_1 \bigg] .
\end{align*}

 Both sides of the equation have a common factor of $\exp(-\rho s)$. 
We can divide both sides by $\exp(-\rho s)$ to simplify.

\begin{equation}
\left\{ 2 \xi \left[ (1 + \phi) \tilde{X}(s) + m(s) \right] + c_0 m(s) \right\}  D_2^2 
= 4  \xi  (1 + \phi)\bigg\{ 2 \xi\exp(-\rho s) (1 + \phi) \left[ (1 + \phi) \tilde{X}(s) + m(s) \right] + D_1 \bigg\} .
\end{equation}
After grouping the terms involving $m(s)$ and the constants 
\begin{equation*}
D_2^2 \left[ 2 \xi (1 + \phi) \tilde{X}(s) + (2 \xi + c_0) m(s) \right]= 4 \xi (1 + \phi) \left[ 2 \xi\exp(-\rho s) (1 + \phi)^2 \tilde{X}(s) + 2 \xi (1 + \phi) m(s) + D_1 \right].
\end{equation*}
Moving all terms involving $m(s)$ to one side and constant terms to the other side, and subtract $ D_2^2 2 \xi (1 + \phi) \tilde{X}$ from both sides:
\begin{equation*}
D_2^2 (2 \xi + c_0) m(s) = 4 \xi (1 + \phi) \left[ 2 \xi\exp(-\rho s) (1 + \phi)^2 \tilde{X}(s) + 2 \xi (1 + \phi) m(s) + D_1 \right] - D_2^2  2 \xi (1 + \phi) \tilde{X}(s).
\end{equation*}

After collecting the terms involving $m(s)$ on the left-hand side and the constant terms on the right-hand side yield
\begin{equation*}
D_2^2 (2 \xi + c_0) m(s) - 4 \xi (1 + \varphi)  2 \xi (1 + \phi) m(s) = 4 \xi (1 + \phi) \left[ 2 \xi\exp(-\rho s) (1 + \phi)^2 \tilde{X}(s) + D_1 \right] - D_2^2 2 \xi (1 + \phi) \tilde{X}(s).
\end{equation*}

Factoring $m(s)$ out and dividing both sides by the coefficient of $m(s)$ implies
\begin{equation}\label{l.1}
m^*(s) = \frac{ \xi (1 + \phi)\bigg[2\xi\exp(-\rho s)(1 + \phi)^2 \tilde{X}(s) +  D_1 - D_2^2 \tilde{X}(s)\bigg]}{D_2^2 (2 \xi + c_0) - 2 \xi^2 (1 + \phi)^{2}}.
\end{equation}

 \section{Simulation Studies:}

 In this section, we conduct a series of simulation studies to evaluate the performance of Equation \eqref{l.1} under varying conditions. To achieve this, we first introduce three distinct models, each characterized by a unique set of parameters, specifically \( c_0, \xi, \phi, \tilde\theta, \sigma \), and \( \rho \). These parameters dictate the behavior of the system and influence the outcomes of our simulations. Once the models are established, we compute the values of \( \tilde m \), which is derived as a logistic transformation of \( m^* \), given by  

\begin{equation}\label{l.2}
\tilde m(s):= \frac{1}{1+\exp\big\{-m^*(s)\big\}}
\end{equation}

This transformation ensures that the complementarity strategy adopted by a firm in response to a monetary shock is confined within the unit interval \([0,1]\), thereby maintaining interpretability and consistency within our framework. Given that the true data-generating process is explicitly determined by Equation \eqref{l.1}, we have the flexibility to simulate arbitrarily large datasets for \( \tilde X \) and systematically analyze the behavior of \( \tilde m \) under each of the three parameter settings. By doing so, we can assess the robustness and sensitivity of the model to different parameter variations, allowing us to infer the conditions under which the complementarity strategy remains stable or fluctuates significantly. The insights gained from these simulation studies will be instrumental in understanding the theoretical and practical implications of Equation \eqref{l.1}, particularly in the context of monetary shocks and firm responses. Through systematic comparisons across different parameter regimes, we can identify key drivers of variation in \( \tilde m \) and evaluate the extent to which the transformation effectively captures the underlying dynamics of the model. Throughout the simulation of three model we choose $\tilde X=0.87$, since we want to observe the effect of higher percent deviation of a firm from CPI on $\tilde m$. Furthermore, in Figure \ref{fig:relations}, $\tilde m$ takes the peak at $\tilde X=0.87$. Three models are as follows:\\
(i). {\bf Model 1:} All the parameter values are same:\\
\begin{equation}\label{l.3}
m^*(s) = \frac{ 0.0001 (1 + 0.0001)\bigg[0.0002\exp(-0.0001 s)(1 + 0.0001)^2 \tilde{X}(s) +  D_1 - D_2^2 \tilde{X}(s)\bigg]}{D_2^2 (0.0003) - 2 (0.0001)^2 (1.0001)^{2}},
\end{equation}
where $\tilde m$ is defined by Equation \eqref{l.2}.\\
(ii). {\bf Model 2:}\\
\begin{equation}\label{l.4}
m^*(s) = \frac{ 0.01 (1.01)\bigg[0.02\exp(-0.01 s)(1.01)^2 \tilde{X}(s) +  D_1 - D_2^2 \tilde{X}(s)\bigg]}{D_2^2 0.03 - 2 (0.01)^2 (1.01)^{2}},
\end{equation}
where $\tilde m$ is defined by Equation \eqref{l.2}.\\
(iii). {\bf Model 3:} All the parameters take distinct higher values: \\
\begin{equation}\label{l.5}
m^*(s) = \frac{ 1.2\bigg[2*0.8*(1.5)^2\exp(-\rho s) \tilde{X}(s) +  D_1 - D_2^2 \tilde{X}(s)\bigg]}{D_2^2 (2*0.8 + 0.8) - 2 *0.8^2 (1.5)^{2}},
\end{equation}
where $\tilde m$ is defined by Equation \eqref{l.2}.\\
Table \ref{tab:dm_test_results} shows three different values of $\tilde m$ for three different sets of parameters.

\begin{table}[H]
\centering
\caption{$\tilde{m}(s)$ test results on simulated datasets.}
\label{tab:dm_test_results}
\begin{tabular}{l|c|c|c}
\hline
 & Model-1 & Model-2 & Model-3 \\
\hline
Predicted $\tilde{m}(s)$ value & 0.617 & 0.396 & 0.047 \\
\hline
$c_0$ & 0.0001 & 0.01 & 0.8 \\
$\xi$ & 0.0001 & 0.01 & 0.8 \\
$\phi$ & 0.0001 & 0.01 & 0.5 \\
$\tilde{\theta}$ & 0.001 & 0.4 & 0.7 \\
$\sigma$ & 0.001 & 0.08 & 0.9 \\
$\rho$ & 0.0001 & 0.01 & 0.8 \\
\hline
\end{tabular}
\end{table}

To further investigate the relationship between parameter values and the complementarity strategy of firms, we systematically varied the parameter sets across the three models and analyzed their corresponding effects on \( \tilde m(s) \). In Model 1, we intentionally selected smaller values for the parameters \( c_0, \xi, \phi, \tilde\theta, \sigma \), and \( \rho \) to observe their impact on the logistic transformation \( \tilde m(s) \). The simulation results indicate that the value of \( \tilde m(s) \) is the highest among the three models, suggesting that lower parameter values yield a stronger complementarity response to monetary shocks. This outcome implies that smaller parameter values make firms more responsive to external shocks, possibly due to reduced inertia or lower adjustment costs. A similar trend was observed in Models 2 and 3, where progressively larger parameter values were assigned to simulate varying firm behaviors. As the parameter values increased from Model 1 to Model 3, the value of \( \tilde m(s) \) consistently decreased, reinforcing the observation that \( \tilde m(s) \) exhibits a negative correlation with the parameter magnitudes. This negative correlation suggests that larger parameter values mitigate the firm's propensity to adopt complementarity strategies, potentially reflecting higher adjustment costs, greater uncertainty, or stronger persistence in firm behavior. The consistent pattern observed across all three models highlights the sensitivity of the complementarity strategy to parameter variations, providing insights into how firms' strategic responses may be influenced by structural characteristics. These findings underscore the importance of parameter selection in modeling firm behavior and suggest that smaller parameter values can amplify the effects of monetary shocks on firm decisions, whereas larger values may dampen the complementarity response.

\begin{figure}[H]
	\includegraphics[width=.98\textwidth]{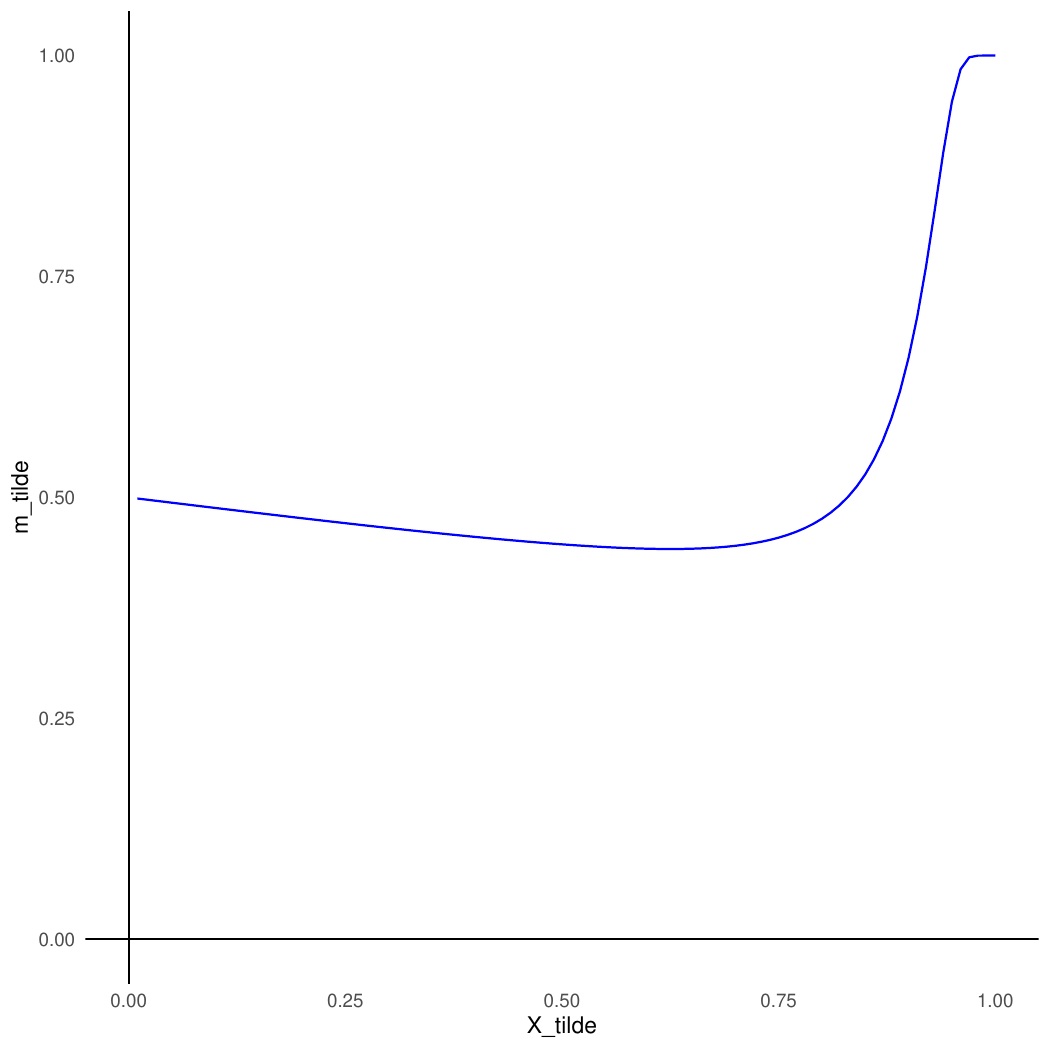}
		\caption{{ Relation ship between $\tilde X$ and adjusted complementarity strategy of a firm due to random monetary shock.}}
		\label{fig:relations}
\end{figure}

Figure \ref{fig:relations} provides a comprehensive illustration of how the variable $\tilde{X}$ interacts with a firm’s adjusted complementarity strategy in response to a random monetary shock. The construction of this figure is based on specific parameter choices, $\xi=0.0001$, $\phi=0.0001$, $c_0=0.00001$, $\rho=0.5$, $\sigma=0.05$, and $\tilde{\theta}=0.02$, ensuring a controlled and precise simulation environment. To perform the simulation, we utilized 100 values of $\tilde{X}$ ranging from 0 to 1. The simulation reveals a crucial insight; at $\tilde{X}=0$, the function $\tilde{m}(s)$ assumes a value of $0.5$. This indicates that when a firm does not exhibit any deviation from the aggregate CPI, its complementarity strategy remains neutral at $0.5$. This neutrality comes from the firm's balanced expectations about the effects of the monetary shock, assigning equal probabilities (ie $1/2$) to potential outcomes. As $\tilde{X}$ increases, we observe a gradual decline in $\tilde{m}$ until it reaches a minimum at $\tilde{X}=0.77$. This decreasing trend suggests that, within the interval $[0,0.77]$, the firm exhibits risk-averse behavior, opting for a more conservative complementarity strategy as it perceives an increased risk in deviating from the aggregate CPI. However, beyond $\tilde{X}=0.77$, a stark shift occurs, $\tilde{m}$ begins to rise sharply. This inflection point signifies a transition in the firm's stance; as $\tilde{X}$ surpasses 0.77, the firm increasingly diverges from the aggregate CPI, signaling a shift towards risk-loving behavior. The sharp upward movement in $\tilde{m}$ aligns with the firm’s growing willingness to take on risk as it perceives larger deviations as more rewarding. This pattern aligns intuitively with our theoretical framework, given that $\tilde{m}$ is derived from a logistic transformation of $m^*(s)$, and the use of an exponential function with a positive coefficient inherently produces this characteristic curve. The interplay between risk attitudes and monetary shocks, as visualized in Figure \ref{fig:relations}, thus underscores the dynamic nature of a firm's strategic adjustments in response to uncertainty, reflecting both risk-averse and risk-loving behaviors at different levels of $\tilde{X}$.

\section{Data Analysis:}

In this section, we will dive into a comprehensive analysis of the dataset, with the aim of uncovering valuable information and identifying significant trends that can inform business strategies and decision-making processes. Data analysis serves as a powerful tool in modern business and economic decision making, as it allows us to transform raw data into meaningful conclusions. By applying various statistical, computational, and visualization techniques, we can explore the dataset, discern patterns, and forecast potential future developments that can have profound implications for business operations and market strategies.

The data set on hand encompasses key economic indicators and financial data that are critical for understanding broader economic conditions and assessing the performance of individual companies. It includes the CPI for 2022 and 2023, which is a key measure of inflation, alongside stock prices for Nestlé, Palmolive, Westrock Coffee, and Dover for the same two years. The CPI data provide a snapshot of the changes in the cost of living over time, reflecting the changes in average prices of a basket of goods and services consumed by households. This data set is essential for understanding inflationary pressures, which can directly affect business costs, consumer purchasing power, and overall economic health. By analyzing the CPI, we can gain insight into how inflation has evolved over the past two years, how it compares to historical trends, and how businesses can adapt their strategies to mitigate the effects of rising or falling prices.

Together with the CPI data, the stock price data for the selected companies—Nestlé, Palmolive, Westrock Coffee, and Dover—will offer valuable insight into how these companies have performed in the market during the same time period. Stock prices reflect investor sentiment, corporate performance, and market dynamics \citep{pramanik2024stochastic}.  This analysis will utilize a combination of descriptive statistics, time series analysis, and data visualizations to uncover trends and potential causes of stock price fluctuations. Descriptive statistics will provide an overview of the central tendencies, variability, and distributions in the CPI and stock price data, while time-series analysis will help us understand how the values have changed over time and identify any seasonality or long-term trends.

In addition to statistical techniques, data visualization will play a critical role in illustrating these trends and relationships. Graphs such as line charts, boxplot , and qqplot will help in visually representing the CPI trends over time, the stock price movements of the companies, and the potential correlations between these variables. By presenting the data visually, we can more easily spot patterns, outliers, and trends that may not be immediately apparent through raw numbers alone.

Through this detailed analysis, our goal is to derive actionable insights that can drive business decisions. By understanding the relationship between inflation (as represented by CPI) and stock performance, companies can adapt their strategies in real-time to navigate fluctuating economic conditions. For instance, companies facing inflationary pressures may consider adjusting their product pricing, refining their marketing strategies, or reevaluating their supply chain management practices \citep{pramanik2025stubbornness}. Furthermore, by examining the stock price performance, businesses can make more informed decisions about investments, mergers and acquisitions, and strategic growth plans.

Ultimately, the insights gleaned from this analysis will enable business leaders and investors to make more informed, data-driven decisions that enhance long-term success. By leveraging both economic and financial data, companies can position themselves more strategically in the market, responding proactively to economic trends and ensuring they remain competitive, resilient, and profitable. This thorough analysis not only benefits companies directly but also provides a deeper understanding of the broader economic landscape, which is essential for shaping the future of business and economic policy .

The CPI is a vital economic metric that tracks the average changes in the prices consumers pay for a specific collection of goods and services. This collection is designed to represent the usual spending habits of households and includes a wide array of items such as food, housing, transportation, clothing, healthcare, and education, among others. Each category in this basket of goods has a designated weight based on its significance to the average household's expenses. For example, housing typically makes up a large portion of the CPI due to its prominence in household budgets \citep{pramanik2025factors}. Its main role is to monitor price fluctuations, a process commonly referred to as inflation. When the CPI rises, it indicates that, on average, the prices of goods and services have increased, leading to a reduction in the purchasing power of money. This phenomenon is referred to as inflation. In contrast, when the CPI declines, it signifies a reduction in prices, known as deflation, which could suggest economic contraction or weakening demand for goods and services \citep{pramanik2025strategies}. Recognizing these shifts is crucial, as inflation or deflation can have profound effects on a nation’s economy, influencing everything from consumer behavior to decisions about monetary policy.

Additionally, the CPI is used to adjust income and expenditure data to account for inflation, thereby ensuring that economic figures remain accurate over time. This adjustment, often called a cost-of-living adjustment, is crucial for maintaining the purchasing power of individuals as prices rise. For example, many social security payments, wages, and pensions are regularly adjusted based on changes in the CPI to ensure they keep pace with rising living costs. Similarly, businesses and policymakers rely on CPI data to understand how price shifts may affect production costs, consumer demand, and overall economic growth.

There are various versions of the CPI. The CPI-U (Consumer Price Index for All Urban Consumers) is the most widely used and includes a broad range of urban households. It tracks price changes for a large segment of the population living in urban areas. Another version, the CPI-W (Consumer Price Index for Urban Wage Earners and Clerical Workers), focuses specifically on households where at least half of the income comes from wages or clerical jobs. A variant known as the Core CPI excludes volatile food and energy prices, providing a clearer view of long-term inflation trends, as food and energy prices can fluctuate significantly in the short term \citep{pramanik2025optimal}. Economists, government agencies, and central banks like the Federal Reserve depend on the CPI to guide monetary policy decisions. For example, when inflation is high, a central bank may raise interest rates to curb borrowing and spending, helping to slow down the economy and bring prices under control. On the other hand, if inflation is low or there is a risk of deflation, central banks may lower interest rates to encourage spending and investment.

While the CPI is an essential tool for understanding inflation, it does have its limitations. One notable drawback is that the basket of goods and services used to compute the CPI may not perfectly reflect the actual consumption patterns of every individual or household. For instance, the spending habits of retirees or low-income families may differ significantly from the average urban consumer. Moreover, the CPI does not adjust for quality changes over time. If a product's quality improves, its price may rise, but this may not necessarily reflect inflation in the traditional sense. Technological advancements and new products, which do not neatly fit into the CPI's basket, can also influence consumer behavior and living standards in ways that the index does not fully capture.

Despite these limitations, the Consumer Price Index remains one of the most significant and widely used economic indicators. It plays a key role in understanding inflation, adjusting wages and income, and informing policy decisions that affect the broader economy. By providing a snapshot of price changes over time, the CPI helps inform not only policymakers but also businesses and consumers, guiding decisions on spending, saving, and investing.
The QQ plot reveals that a majority of the data points closely follow the theoretical quantile line, which strongly indicates that the CPI data adheres to a normal distribution. This graphical representation is a powerful tool for assessing the distribution of a dataset, and in this case, the alignment of the data points with the straight line suggests that the data does not deviate significantly from the expected normal pattern. In statistical analysis, a QQ plot is used to compare the quantiles of the sample data with the quantiles of a theoretical distribution—in this case, the normal distribution. When the data points lie near or along the line, it demonstrates that the sample data approximates the normal distribution well.

Such an alignment has important implications for statistical analysis. A normal distribution is a common assumption in many statistical techniques, as it enables the use of various tools like parametric tests, regression analysis, and the calculation of confidence intervals, all of which rely on the data following a normal distribution. This means that, based on the QQ plot, we can confidently assume that the CPI data is normally distributed. As a result, we are in a favorable position to apply statistical methods that presuppose normality, which can help in making predictions, testing hypotheses, and deriving more precise conclusions about the behavior of inflation or price changes over time.

The significance of confirming normality cannot be overstated, especially in economic analysis, where assumptions about data distribution guide critical decision-making processes. With the assumption of normality validated by the QQ plot, the CPI data can be subjected to a range of parametric statistical methods. These methods are advantageous because they tend to be more powerful and yield more accurate results when the data is normally distributed. Therefore, understanding that the CPI data follows this distribution opens up a broader array of analysis techniques that can provide deeper insights into inflation trends, economic forecasting, and policy recommendations. Furthermore, confirming normality through visual tools like the QQ plot adds an extra layer of confidence to our interpretation of the data. statistical tests such as the Shapiro-Wilk test  for normality was used to confirm normality.

\begin{figure}[H]
	\includegraphics[width=.98\textwidth]{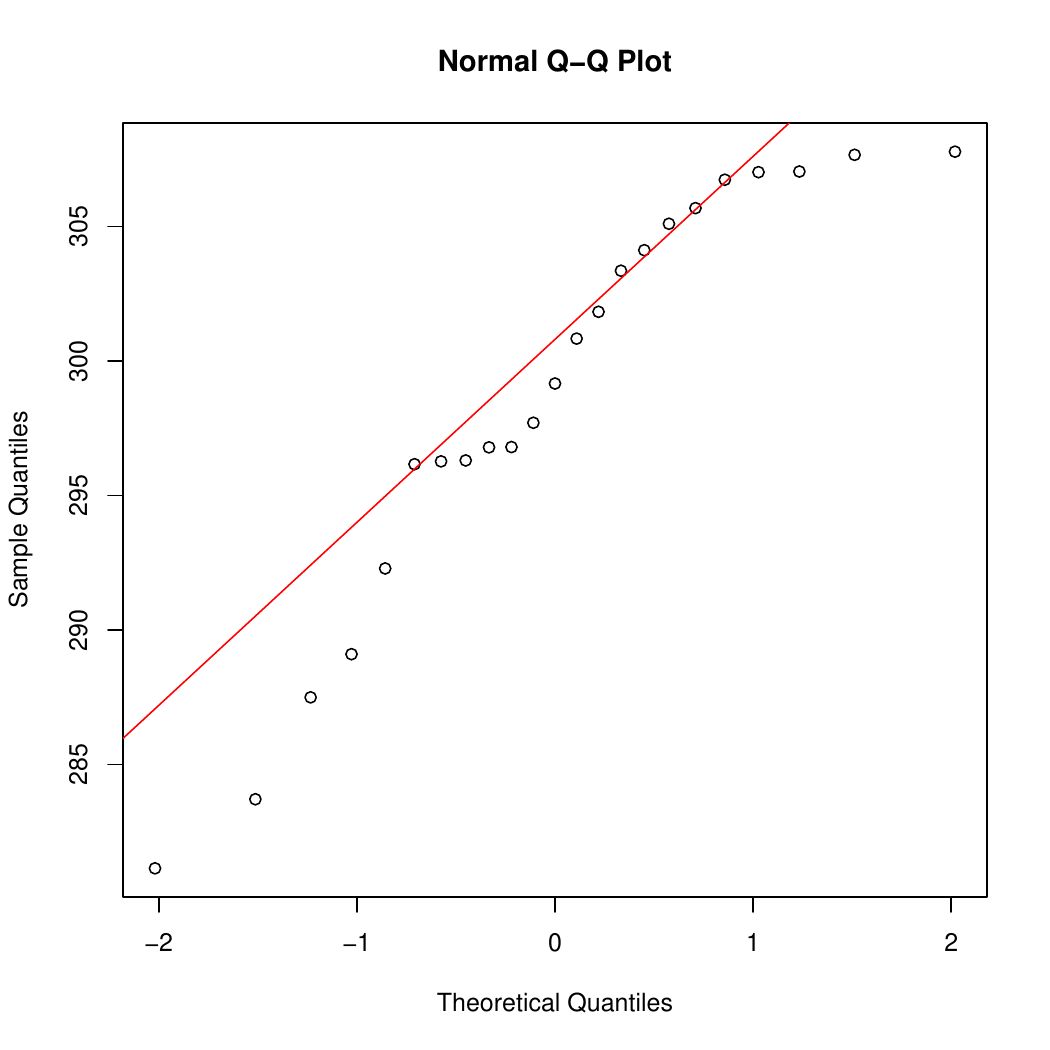}
		\caption{{ Test for normality of CPI.}}
		\label{fig:0.1}
\end{figure}

Observation from the data is that Palmolive exhibits the highest level of variability among the four companies, as evidenced by its box plot (\ref{fig:0}), which is the largest in size compared to the others. This suggests that the measured variable fluctuates more widely for Palmolive, indicating greater dispersion in its data points. The interquartile range (IQR), which measures the spread of the middle 50\% of the data by calculating the difference between the first quartile (Q1) and the third quartile (Q3), is notably wider for Palmolive than for the other companies. A larger IQR signifies that Palmolive’s data distribution is more spread out, meaning there is a broader range of values between its lower and upper quartiles. This level of variation can have significant implications, as it may reflect greater diversity in performance, inconsistent outcomes, or increased sensitivity to external factors affecting the measured variable. In addition to exhibiting the highest variability, Palmolive also has the highest median value among the four companies. The median, which represents the middle value of the dataset when arranged in ascending order, serves as a key indicator of central tendency and provides a robust measure of typical performance, as it is not as affected by extreme values as the mean. The fact that Palmolive’s median is higher than that of its competitors suggests that, on average, it outperforms the other companies in terms of the measured variable. This finding is particularly noteworthy, as it indicates that despite having greater variability, Palmolive still maintains a stronger central performance, making it a standout in comparison to the other companies. Such insights are valuable in understanding market positioning and competitive advantages, as they highlight both consistency and relative superiority in the measured metric. By considering these factors, businesses can better assess how variations in performance might impact strategic decision-making and identify potential areas for improvement or optimization.

\begin{figure}[H]
	\includegraphics[width=.98\textwidth]{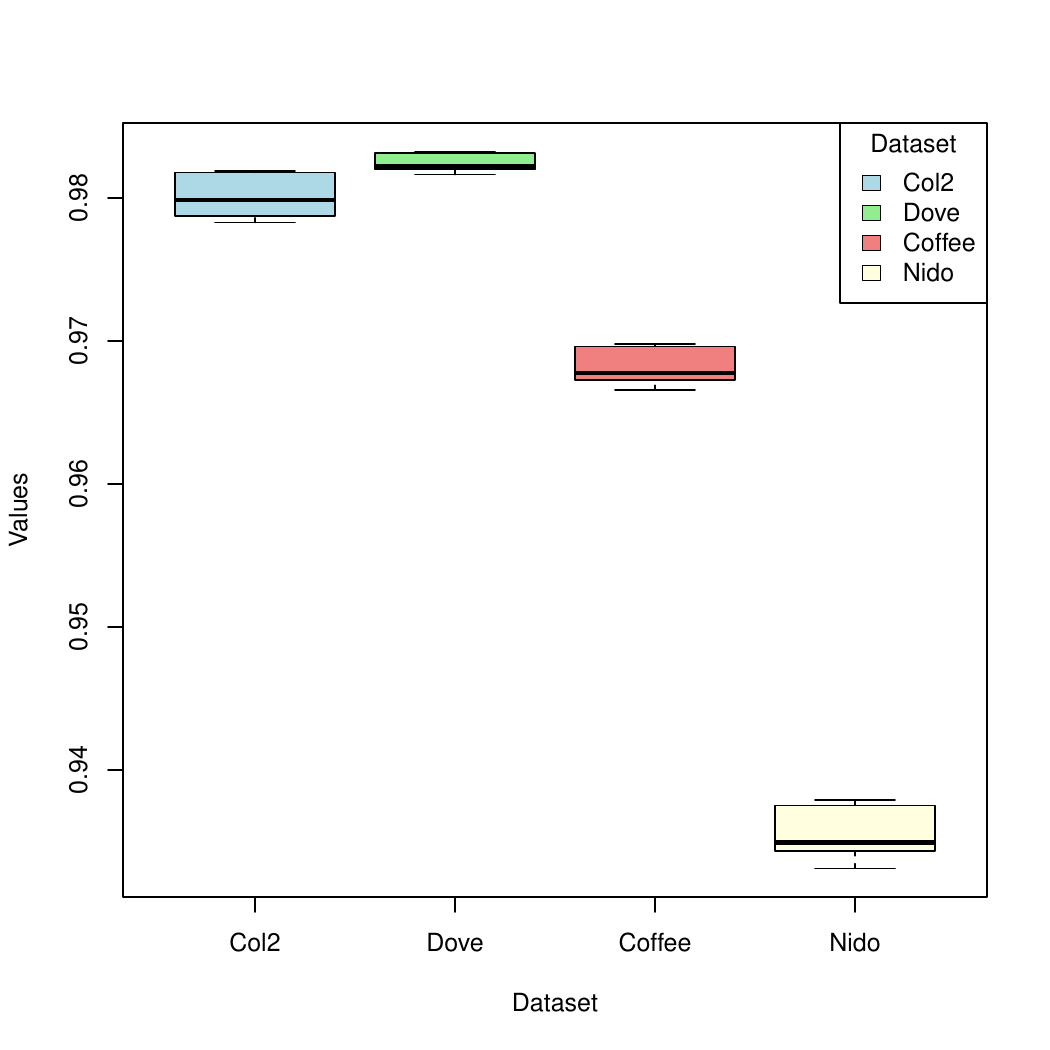}
		\caption{{Box plot of four different firms.}}
		\label{fig:0}
\end{figure}

 The combination of a large box plot and a high median suggests that Palmolive has both high performance and high variability. This could mean that while Palmolive performs well on average, there is significant fluctuation in its performance. Nestle has the second-largest box plot and median. This indicates that Nestle also has relatively high variability and strong performance, but not as high as Palmolive. Nestle's performance is consistent and strong, though it exhibits less variability compared to Palmolive. This suggests that Nestle is a stable performer but may not reach the same peaks as Palmolive. WestRock has a smaller box plot and median compared to Palmolive and Nestle. This indicates lower variability and moderate performance. WestRock's performance is more consistent and less variable than Palmolive and Nestle, but it also has a lower median value. This suggests that WestRock is a stable but less outstanding performer compared to the top two companies.

 Dover has the smallest box plot and the lowest median among the four companies. This indicates low variability and lower performance.
 Dover's performance is the least variable and also the lowest among the four companies. This suggests that Dover is a consistent but underperforming company compared to the others \citep{dong2024strategic}. Palmolive stands out as the top performer with the highest median and the largest variability. This suggests that while Palmolive performs well on average, its performance is less predictable. Nestle follows closely, with strong performance and moderate variability. It is a stable and reliable performer. WestRock shows moderate performance with low variability, indicating consistency but not outstanding results. Dover has the lowest performance and the least variability, suggesting it is the least competitive among the four companies. Dover's line reaches the lowest peak among the four companies, indicating that it achieved the lowest maximum value during the observed time period.Dover may have underperformed relative to the other companies, either due to weaker performance, less growth, or more challenges in maintaining higher values, and its trend line is relatively flat,indicating Dover's performance may be stable but consistently low.

\begin{figure}[H]
	\includegraphics[width=.98\textwidth]{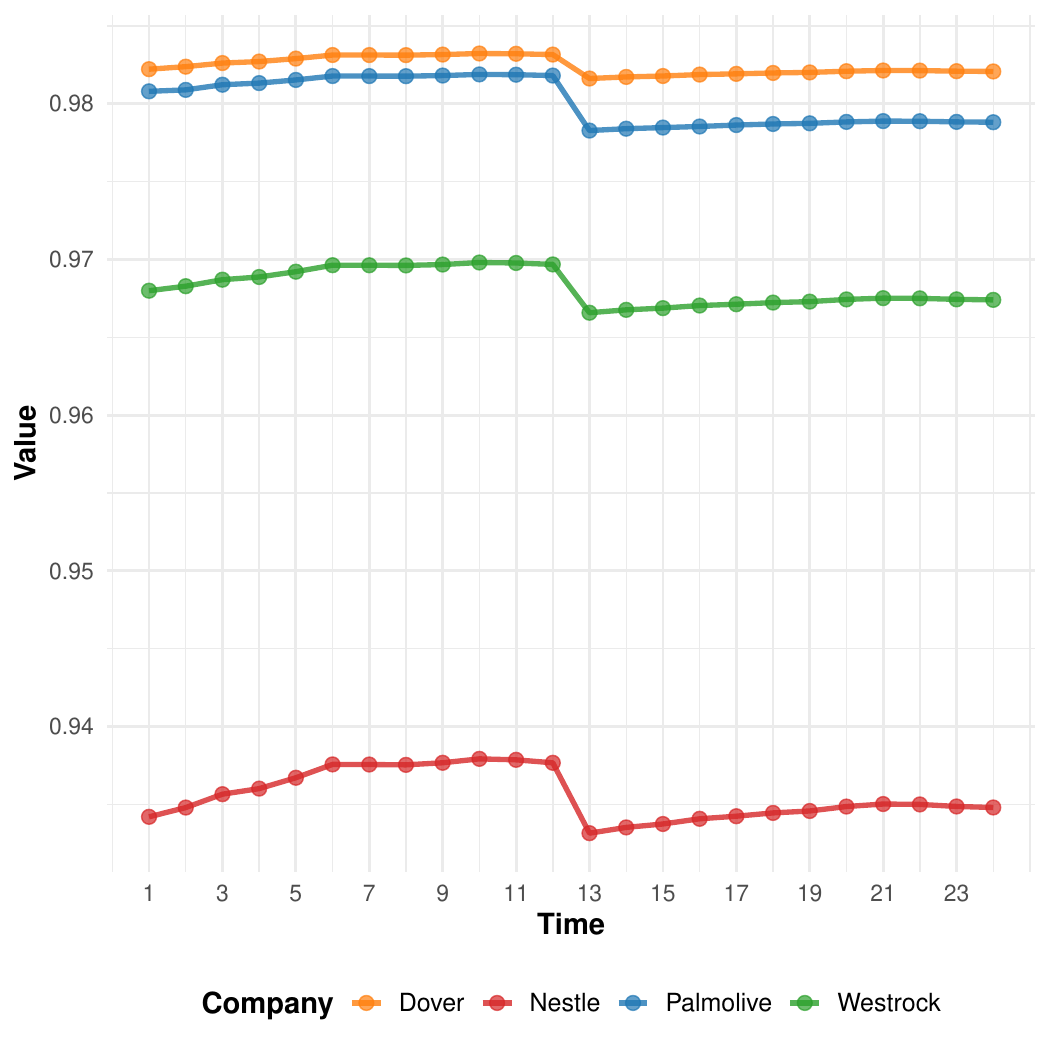}
		\caption{{Line plot of four diffirent firms.}}
		\label{fig:1}
\end{figure}

Westrock's peak is higher than Dover's but lower than Palmolive's and Nestle's, placing it in the second-lowest position and its performance is better than Dover's but still lags behind Palmolive and Nestle.Westrock may be experiencing periodic challenges or variability in its performance because of its periodic fluctuation.

Palmolive's peak is higher than both Dover and Westrock but lower than Nestle's, placing it in the second-highest position.Palmolive is performing better than Dover and Westrock, but is still outperformed by Nestle.The trend  line of Palmolive rises steadily, indicating Palmolive is steadily improving and closing the gap with Nestle.

Nestle's line reaches the highest peak among the four companies, indicating that it achieved the highest maximum value during the observed time period. Nestle is the top performer among the four companies, achieving the highest values at its peak before a sharp decline. Nestle may still be achieving high values but with some variability in performance.
The line plot reveals a clear hierarchy in performance among the four companies, with Nestle leading, followed by Palmolive, Westrock, and Dover. 
The trends over time provide additional insights into how each company is performing and where they may need to focus their efforts to improve or maintain their positions.

In considering the stability and behavior of a system, Westrock provides an insightful framework that can be interpreted through the lens of dynamic systems where equilibrium plays a central role. The primary feature of this system lies in its overall stability, where fluctuations, though present, are generally moderate and do not significantly disrupt the system's core functions or direction. This is not to say that the system is impervious to change or unaffected by random disturbances; rather, it maintains a level of resilience that allows it to absorb these variations without veering too far from its established state.
Westrock's system offers a compelling model for environments where stability and predictability are paramount but still acknowledges the ever-present influence of randomness. While the system is not immune to fluctuations, it demonstrates a natural resilience, remaining anchored in a stable pattern that allows it to continue functioning effectively. The occasional, moderate fluctuations serve more as temporary disturbances, adding a layer of complexity but not undermining the overall order that governs the system’s trajectory. Thus, while fluctuations may momentarily cause ripples, the system ultimately retains its stability, making it highly adaptable to minor, random shifts without experiencing significant upheaval. This kind of stability is characteristic of well-managed environments that balance predictability with the inevitable uncertainties of the external world.
Palmolive, as a brand and a company, operates within an environment that is highly volatile and sensitive to fluctuations, particularly those driven by external factors. This volatility manifests in the form of large, sudden shifts or "jumps" in its market behavior, where the company's performance can change dramatically in a short period of time. The nature of these shifts is not only abrupt but also often unpredictable, making it difficult to forecast or mitigate the impact of such fluctuations. These sudden changes in the system can be attributed to a variety of external shocks that disrupt the status quo and provoke swift reactions from investors, consumers, and industry players alike.

One key factor contributing to this volatility is the sensitivity of the market to global economic events. For instance, during times of financial crisis or economic downturns, the market can experience steep declines as consumer spending drops, production costs rise, or general confidence in the economy falters. These external shocks may have cascading effects on companies like Palmolive, where fluctuations in raw material costs or shifts in consumer preferences can translate into sudden and drastic changes in stock prices, revenue projections, and even operational strategy. The sudden drop in demand for discretionary products during an economic slowdown or a global recession could cause Palmolive to face significant challenges, as the brand relies heavily on consumer behavior patterns that are often affected by economic conditions.

Moreover, Palmolive’s exposure to geopolitical events or market sentiment shifts further amplifies the company's inherent volatility. For example, a sudden change in trade policy, such as tariffs or sanctions, can disrupt the supply chain, affecting the availability and cost of materials. Similarly, an unanticipated regulatory change in one of the regions where Palmolive operates could lead to an immediate and dramatic shift in the company’s operations, requiring rapid adaptation to new rules or standards. These unpredictable events result in sudden ``jumps,” wherein the market reacts sharply, either positively or negatively, to news or rumors, creating instability in stock prices and public perception of the company.

\begin{figure}[H]
	\includegraphics[width=.98\textwidth]{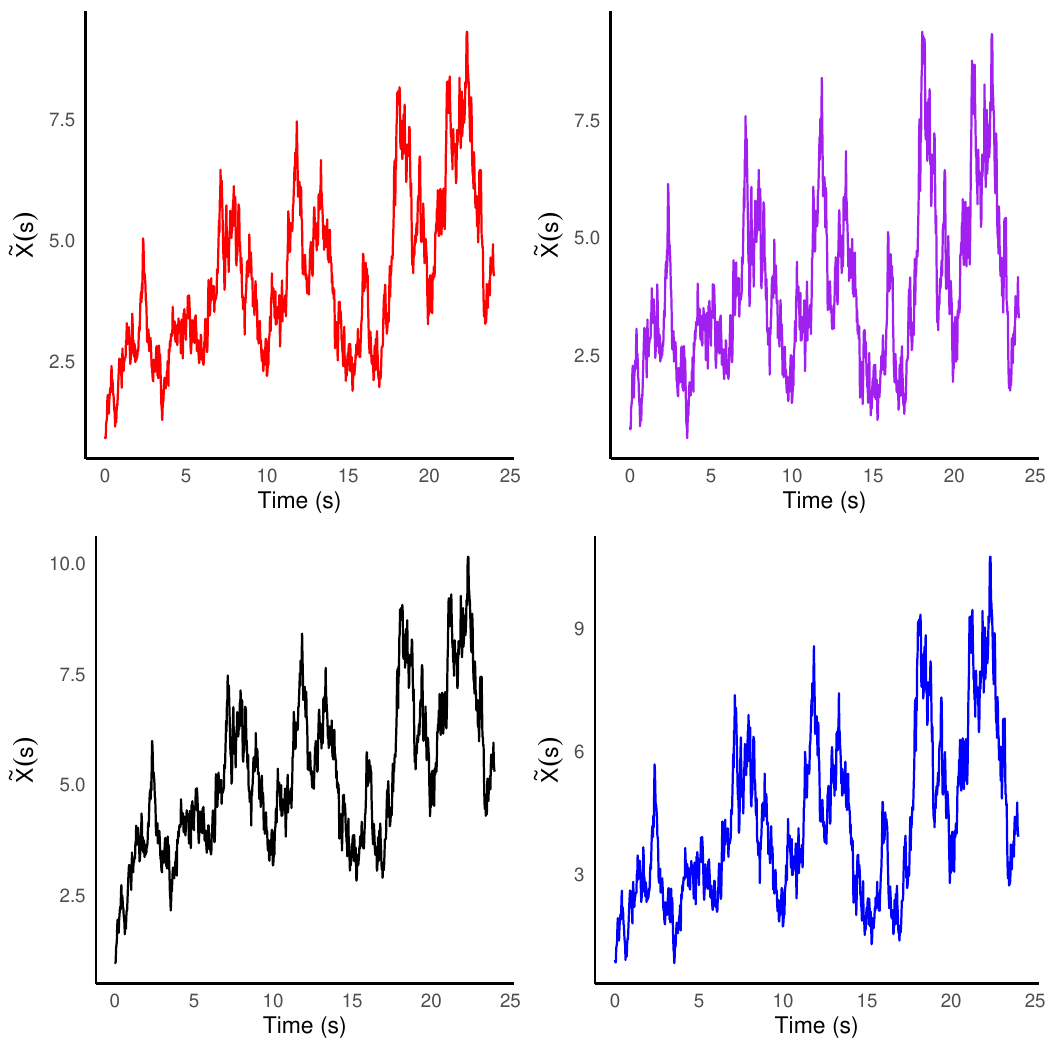}
		\caption{{Stochastic plot of four differnt firms.}}
		\label{fig:2}
\end{figure}

In addition to these external shocks, Palmolive is also vulnerable to shifts in consumer sentiment, which can be influenced by anything from viral social media campaigns to broader changes in lifestyle preferences. For instance, if a trend emerges that encourages consumers to shift towards more sustainable or eco-friendly products, Palmolive might experience a sudden surge in demand—or conversely, if a controversy arises surrounding the brand or its practices, it could face a sharp and immediate decline in consumer trust. These kinds of sentiment-driven fluctuations are often unpredictable, and the market can respond with extreme variations, sometimes amplifying the company's perceived vulnerability.

The degree of variation within Palmolive's system is not simply moderate or minor but is instead extreme, with significant swings occurring with little warning. The company's stock may fluctuate dramatically, as investors attempt to gauge the impact of an event that could affect the broader market, forcing large-scale buying or selling decisions. This results in a heightened sense of instability, where the company's future prospects appear far more uncertain than in more stable systems. Unlike companies with more predictable growth patterns, Palmolive operates in a space where the uncertainty is pronounced, and the future trajectory is often clouded by these external shocks.

These large fluctuations can be triggered by any number of unexpected events. A major news event, such as a natural disaster affecting one of Palmolive’s manufacturing locations, could instantly cause the company’s stock price to plummet as investors factor in potential disruptions to production and distribution. Similarly, geopolitical tensions or an unforeseen political crisis could introduce uncertainty into the company’s international operations, forcing investors to react rapidly and sending shockwaves through the market. Because Palmolive operates in a global marketplace, the interconnectedness of these events means that an incident in one part of the world can cause ripple effects across the entire business, further intensifying the volatility.

Such irregular and often dramatic fluctuations make Palmolive’s market environment a challenging one for investors, who must constantly monitor both macroeconomic indicators and micro-level consumer behaviors in order to stay ahead of the curve. Unlike a stable system where gradual changes over time are the norm, Palmolive’s system requires constant vigilance, as the risk of sudden, pronounced fluctuations is always present. The company must remain agile and responsive to the volatile nature of the market, preparing for the possibility of unexpected shifts that could dramatically alter its trajectory at any given moment.

In conclusion, Palmolive is situated in a highly volatile system characterized by frequent, large variations in its market performance. This volatility arises from external shocks such as economic crises, geopolitical tensions, and shifts in consumer sentiment. These unpredictable factors introduce instability into the system, leading to significant jumps in the company’s value and performance. The degree of variation is extreme, and the overall market environment for Palmolive is one where the future is uncertain, requiring the company to remain adaptable and prepared for sudden and often unpredictable changes. For investors and stakeholders, this volatility presents both risk and opportunity, making it essential to closely monitor the broader forces at play in the global marketplace.
Dover, as a company and system, operates in an environment that is characterized by a high degree of predictability and stability. Unlike markets that experience extreme fluctuations or those that are heavily influenced by external disruptions, Dover follows a consistent, deterministic trend over time. This means that the company’s trajectory can be anticipated with a high level of confidence, as it generally moves in a single direction without significant deviations. Such predictability allows for long-term planning and decision-making, as stakeholders can reasonably expect that the system will continue along a familiar path, barring any extraordinary circumstances.

This stability is reflected in the system's behavior, where fluctuations are not only minor but almost negligible. Unlike industries where market dynamics lead to constant shifts in product demand, regulatory changes, or consumer sentiment, Dover operates in a relatively calm and steady environment. The factors that typically drive volatility in other sectors, such as external shocks or rapid technological innovations, have a much smaller impact on Dover’s operations. Instead, the company’s performance is shaped by more predictable, consistent forces—such as steady demand for its products, ongoing improvements in operational efficiency, and a stable industry environment.

As a result of this predictable, deterministic nature, the system experiences little noise. In market terms, “noise” refers to the random fluctuations or unpredictable events that can obscure underlying trends. For Dover, such noise is virtually absent. The company's steady progression can be likened to a well-oiled machine where each component functions with precision and minimal disruption. Whether it’s production schedules, product development, or market growth, Dover's performance exhibits a sense of orderliness that reduces the likelihood of unexpected deviations. Investors, customers, and employees alike can rely on this steady state, as the company’s operations remain largely unaffected by the usual volatility that can plague other industries.

In a system like Dover’s, where there are minimal disruptions, there are also no significant jumps. These “jumps” refer to sudden, sharp changes that could drastically alter the company’s market position or operational strategy, often due to external factors like a crisis, a sudden shift in consumer preferences, or a technological breakthrough. In Dover’s case, such abrupt shifts are uncommon. The company’s steady growth is typically the result of well-considered strategies and incremental progress, rather than drastic, sudden changes. As a consequence, there’s an inherent tranquility in the company’s operations, and stakeholders can rely on the fact that, barring unforeseen factors, the company’s path forward is both clear and continuous.

The absence of sharp, sudden shifts makes Dover’s system incredibly predictable. For investors and decision-makers, this means that the company’s future is largely known, based on historical trends and well-understood market dynamics. In industries marked by volatility and uncertainty, this stability is a significant advantage. Whether in terms of forecasting earnings, assessing market share, or projecting growth, Dover’s relatively stable trajectory ensures that these predictions can be made with a high degree of certainty. Unlike companies in more volatile sectors, where the future is often clouded by unpredictable forces, Dover’s consistency in performance makes it a reliable player in its industry.

What’s more, Dover’s ability to follow a clear, deterministic pattern allows for long-term planning without the constant need for reactive adjustments. Companies in unpredictable environments must often pivot in response to sudden changes in market conditions or regulatory environments. However, Dover's predictable trajectory means that the company’s leadership can focus on maintaining and optimizing a proven strategy, refining processes and enhancing efficiencies with the confidence that the broader market conditions will not disrupt their progress in the short term.

The system is not subject to abrupt disruptions; instead, it follows a measured, calculated course that ensures both growth and sustainability. The market environment in which Dover operates is conducive to this type of stability—there are few forces that can alter its operations in any significant way. Whether the company is expanding its product line, entering new markets, or investing in new technologies, these decisions are typically built upon a foundation of steady progress rather than reactive adaptations to market shifts.

In sum, Dover’s system exemplifies stability and predictability in a way that few other companies can replicate. With minimal fluctuations, little noise, and no significant jumps, it provides a clear and steady pattern that stakeholders can depend on. This consistency offers a sense of security and confidence, as the company’s direction remains unmistakably clear. By avoiding sharp disruptions and maintaining a steady path, Dover ensures that its growth is gradual and sustainable, making it a reliable player in its industry over time. Its ability to maintain this kind of calm and order in an often unpredictable world is a testament to its strong foundation and long-term strategic vision, allowing it to thrive without the upheaval that typically marks more volatile systems.
Nestlé, as a corporation, operates within a complex and often unpredictable environment where its behavior can appear erratic and subject to sudden, unforeseen shifts. Unlike systems that follow more structured, deterministic paths, Nestlé’s performance is heavily influenced by a variety of factors that introduce a high degree of uncertainty and volatility into its operations. This unpredictability often manifests in abrupt, irregular changes that do not adhere to any clear pattern or trajectory. The company's market behavior, for instance, can shift dramatically in response to external events, making it difficult to anticipate its future direction with any degree of confidence.

One of the primary characteristics of Nestlé's environment is its vulnerability to large, unpredictable fluctuations. These fluctuations can vary in size and frequency, making the system difficult to predict and challenging for investors to assess. During periods of extreme market speculation or when the economy is subjected to unexpected external shocks such as a financial crisis, a natural disaster, or political instability- Nestlé’s performance can fluctuate widely. These shifts may not follow a logical or smooth progression but instead appear as erratic movements that reflect the broader unpredictability of the external environment. Just like the stock market during times of heightened speculation, Nestlé can experience significant, sudden swings, which do not necessarily align with the company’s long-term strategy or performance indicators.

What makes this system particularly chaotic is the lack of a consistent pattern in its fluctuations. While some industries or companies may exhibit relatively stable, smooth growth trajectories, Nestlé operates in an environment where fluctuations are large and unpredictable. The company could experience periods of relative calm, only to be followed by sharp, drastic changes that push the company’s stock prices or operational focus in completely new directions. These fluctuations may vary from being relatively minor shifts to more significant upheavals that could drastically alter the company’s trajectory. This inherent variability in fluctuation size and frequency suggests that Nestlé’s market behavior is driven by a wide array of dynamic and often uncontrollable factors, which makes it challenging to predict outcomes with accuracy.

One of the most defining aspects of this erratic system is the way in which sudden, drastic shifts occur with little warning. These abrupt changes often serve as a signal that external forces have interfered with the company’s internal mechanisms, knocking it off its predictable course. For instance, changes in government regulations, shifts in consumer behavior, supply chain disruptions, or even global geopolitical events can trigger sharp reactions within Nestlé’s operations, creating an environment where stability becomes fleeting, and the company’s future becomes highly uncertain. These “jumps” in the system can be quite drastic, representing significant changes in business operations, financial performance, or market strategy.

\begin{figure}[H]
	\includegraphics[width=.98\textwidth]{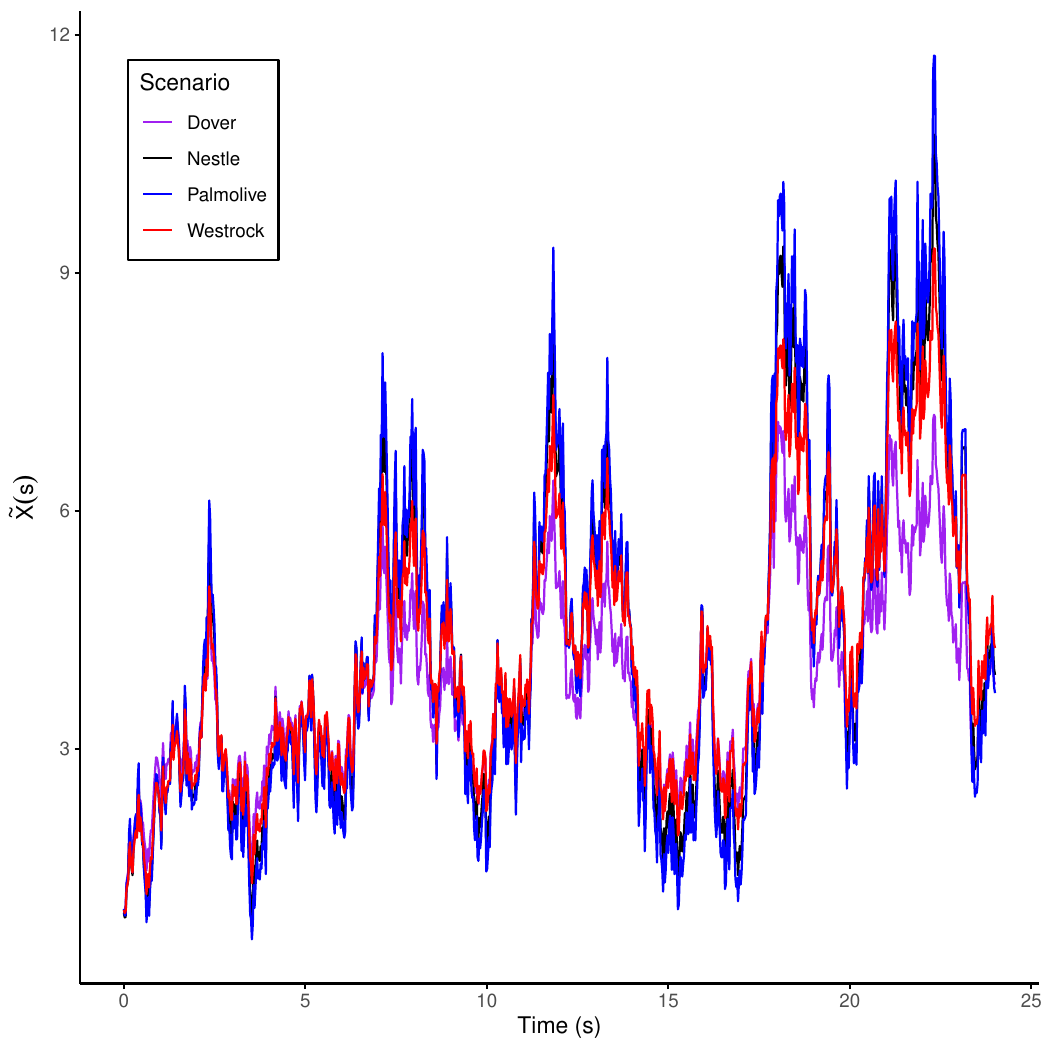}
		\caption{{Combined stochastic plot of four different firms.}}
		\label{fig:3}
\end{figure}

This volatility is often the result of Nestlé’s dependence on external conditions, which are themselves subject to rapid change. The global food industry, in particular, is vulnerable to a wide range of factors that can lead to instability, including shifts in commodity prices, changing consumer preferences, political unrest, and environmental factors such as climate change. These unpredictable forces can disrupt Nestlé’s carefully laid plans, causing the company to react suddenly and adapt to new conditions. Consequently, the company finds itself navigating through a landscape that is rarely calm, where any number of unpredictable events could dramatically alter the company's prospects at a moment’s notice.

Furthermore, the lack of a predictable trend means that Nestlé’s performance may react abruptly to any number of events, creating an environment that is hard to navigate for both investors and company leaders. For example, a new competitor may emerge with a disruptive product, or a global recession could hit, triggering a rapid shift in consumer demand. On top of this, shifts in investor sentiment can create swings in Nestlé’s stock price, compounding the unpredictability and creating a volatile financial landscape. The combination of internal and external factors contributes to a system that is highly sensitive to change, where small triggers can lead to disproportionately large consequences.

In such a chaotic system, it is difficult for any long-term predictions to hold much weight. The system’s inherent volatility means that traditional forecasting models or trend analysis may be less effective, as the company's performance can change in unexpected ways without any discernible pattern to follow. Therefore, Nestlé’s management must be highly adaptable, responding quickly to external events while attempting to mitigate risks as best as possible. Despite the challenges posed by this erratic behavior, the company must maintain flexibility, ready to pivot its operations or adjust its market strategy in response to sudden shifts.

In conclusion, Nestlé operates in a highly erratic and volatile system, where large, unpredictable fluctuations are the norm rather than the exception. The company’s performance is subject to abrupt shifts that disrupt any clear or consistent trend, creating a chaotic and unpredictable environment. The system is frequently knocked off balance by external shocks, forcing the company to react quickly to these changes. In such a landscape, stability becomes an elusive goal, and the ability to navigate uncertainty and adapt to sudden shifts is crucial for Nestl\'e’s survival and growth. The erratic nature of its operations makes it a dynamic but challenging system to manage, where the outcome of any given scenario is highly uncertain and subject to a wide range of influencing factors.
\begin{table}[H]
    \centering
    \caption{Estimated control values ($\tilde{m}$) for different companies and models based on six parameters. The last row represents the overall $\tilde{m}$ values for each company and model.}
    \resizebox{\textwidth}{!}{ 
        \begin{tabular}{lccc ccc ccc ccc}
            \toprule
            & \multicolumn{3}{c}{Nestle} & \multicolumn{3}{c}{Westrock} & \multicolumn{3}{c}{Dover} & \multicolumn{3}{c}{Palmolive} \\
            \cmidrule(lr){2-4} \cmidrule(lr){5-7} \cmidrule(lr){8-10} \cmidrule(lr){11-13}
            Parameter & Model 1 & Model 2 & Model 3 &  Model 1 & Model 2 & Model 3 &  Model 1 & Model 2 & Model 3 &  Model 1 & Model 2 & Model 3 \\
            \midrule
           $c_0$ & $0.001$ & $0.01$ & $0.8$ & $0.001$ & $0.01$ & $0.8$ & $0.001$ & $0.01$ & $0.8$ & $0.001$ & $0.01$ & $0.8$ \\
            $\xi$ & $0.001$ & $0.01$ & $0.8$ & $0.001$ & $0.01$ & $0.8$ & $0.001$ & $0.01$ & $0.8$ & $0.001$ & $0.01$ & $0.8$ \\
            $\phi$ & $0.001$ & $0.01$ & $0.5$ & $0.001$ & $0.01$ & $0.5$ & $0.001$ & $0.01$ & $0.5$ & $0.001$ & $0.01$ & $0.5$ \\
            $\tilde\theta$ & $0.001$ & $0.4$ & $0.7$ & $0.001$ & $0.4$ & $0.7$ & $0.001$ & $0.4$ & $0.7$ & $0.001$ & $0.4$ & $0.7$ \\
            $\sigma$ & $0.001$ & $0.08$ & $0.9$ & $0.001$ & $0.08$ & $0.9$ & $0.001$ & $0.08$ & $0.9$ & $0.001$ & $0.08$ & $0.9$ \\
            $\rho$ & $0.001$ & $0.01$ & $0.8$ & $0.001$ & $0.01$ & $0.8$ & $0.001$ & $0.01$ & $0.8$ & $0.001$ & $0.01$ & $0.8$ \\
            \midrule
             $\tilde{m}$ & $0.883$ & $0.236$ & $0.199$ & $0.999$ & $0.940$ & $0.323$ & $0.988$ & $0.298$ & $0.278$ & $0.999$ & $0.342$ & $0.052$ \\
            \bottomrule
        \end{tabular}
    }
    \label{tab:control_values}
\end{table}
Table \ref{tab:control_values} provides a comprehensive overview of the estimated values for the complementarity strategy of a firm due to random monetary shock ($\tilde{m}$) for four major companies—Nestle, Westrock, Dover, and Palmolive—across three different models. Each model is characterized by six key parameters: $c_0$, $\xi$, $\phi$, $\tilde{\theta}$, $\sigma$, and $\rho$. The structure of the table is designed to allow for a direct comparison between the companies and their respective models, facilitating the evaluation of how each parameter varies across different configurations. Each company has three columns corresponding to the three models, enabling a side-by-side examination of the estimated values under varying conditions. 

The values in table \ref{tab:control_values} indicate that certain parameters remain consistent across companies and models, such as $c_0$, $\xi$, and $\rho$, all of which take the values 0.001, 0.01, and 0.8, respectively, for Model 3 across all companies. However, parameters such as $\tilde{\theta}$ and $\sigma$ show more variation, particularly in Model 2 and Model 3, suggesting a stronger influence on the estimated complementarity strategy. The last row of the table provides the overall values of the $\tilde{m}$ for each company, summarizing the combined effect of all parameters within each model. 

The overall values differ significantly across models and companies, highlighting potential discrepancies in model performance or differences in company-specific characteristics. For instance, while Nestle’s overall complementarity strategy values range from 0.199 to 0.883 across models, Westrock exhibits a higher value for Model 1 (0.999) but a substantially lower value for Model 3 (0.323), indicating possible variations in parameter sensitivity or effectiveness across different modeling approaches. Similarly, Palmolive demonstrates the most substantial variation, with Model 3 producing an exceptionally low complementarity strategy value of 0.052 compared to the higher estimates in the other models. This suggests that Palmolive’s response to different model conditions might be more sensitive or significantly different from the other companies. 

\section{Discussions:}
In this paper, we establish a robust framework for determining the optimal complementarity strategy of a firm operating under a CIR stochastic process with a non-zero drift, providing a novel extension to existing models of firm decision-making under uncertainty. The core objective of the firm is to maximize its payoff function while accounting for percent deviations from symmetric equilibrium in both its own pricing strategy and the aggregate CPI, which are treated as the state variables in our formulation. The firm’s response to random monetary shocks, which dictates its complementarity strategy, serves as the principal control variable in our optimization process. 

Unlike classical models that often assume a driftless stochastic environment, our approach integrates a non-zero drift component, thereby allowing for a more realistic representation of economic dynamics where inflationary trends or persistent shocks influence firm behavior over time. To track this problem, we extend the widely used an extended version of Calvo model by incorporating a mean field approach, which enables us to derive a closed-form solution for the firm’s optimal complementarity strategy. This analytical framework not only enhances our understanding of firm behavior in stochastic environments but also provides valuable insights into the role of market volatility in shaping strategic decision-making. 

Theoretical findings in this paper suggest that as volatility in the economic environment increases, firms tend to reduce their complementarity strategies, indicating a heightened preference for risk aversion and a more cautious approach to pricing adjustments in response to external shocks. To test the validity of our theoretical predictions, we apply the model to four major consumer goods firms: Nestle, Westrock, Dover, and Palmolive. The empirical analysis, based on table \ref{tab:control_values}, reveals that the reduction in complementarity strategy due to increasing market volatility is far more pronounced than what is predicted by the theoretical model. This deviation underscores the complex nature of real-world market uncertainty, suggesting that factors beyond those captured in the CIR stochastic framework—such as behavioral biases, supply chain frictions, or competitive market pressures—may further exacerbate the decline in complementarity strategies. These findings have significant implications for firm strategy and economic policy, as they indicate that firms facing high levels of market volatility may adopt defensive pricing behaviors that further influence aggregate economic stability. Furthermore, our study highlights the need for refining theoretical models to better account for real-world deviations and underscores the importance of empirical validation in economic modeling.

The framework we examine plays a crucial role in advancing the study of equilibrium dynamics, serving as a foundation for further exploration in related literature. Notably, \cite{alvarez2021empirical} leveraged the equilibrium characterization first introduced in \cite{alvarez2023price} to analyze the impulse response of an economy to shocks containing a transitory component, marking a departure from the conventional focus on one-time, permanent shocks that dominate much of the existing research. By incorporating transitory elements, their analysis offered a more nuanced perspective on how economic agents dynamically adjust their decisions in response to short-lived disturbances, providing insights that are particularly relevant for understanding real-world price-setting behavior and policy implications. Building on this foundational work, our study extends the framework further by exploring the role of higher-order perturbations, a methodological innovation that allows for a deeper examination of nonlinearities and complex interactions within the equilibrium structure. This extension is particularly significant for the ongoing debate between time-dependent and state-dependent models, as it provides a rigorous way to differentiate their responses to shocks of varying magnitudes. In traditional settings, small shocks tend to be absorbed more smoothly, with responses that may be adequately captured by linear approximations, whereas larger shocks can elicit nonlinear dynamics that fundamentally alter equilibrium behavior. Our approach seeks to systematically quantify these differences, shedding light on the mechanisms that drive divergent responses in time- and state-dependent models when subjected to varying shock sizes. By doing so, we aim to provide a more comprehensive analytical framework that enhances the robustness of comparative studies in macroeconomic modeling, particularly in areas where policymakers must navigate environments characterized by both frequent small disturbances and occasional large shocks. The ability to distinguish between these different shock responses is essential for refining theoretical predictions, improving empirical validation, and ultimately designing more effective economic policies that account for the full range of equilibrium dynamics observed in real-world markets.

\section*{Appendix:}

\subsection*{Proof of Lemma \ref{l0}.}
(i). As $\mathcal W(t)-\mathcal W(s)$ is independent of $\mathcal F_s^\mathcal W,\ \forall s\leq t$,
\[
\E\left\{\mathcal W(t)-\mathcal W(s)\bigg|\mathcal F_s^\mathcal W\right\}=\E\{\mathcal W(t)-\mathcal W(s)\}=0.
\]
Hence, $\E\{\mathcal W(t)\big|\mathcal F_s^\mathcal W\}=\mathcal W(s)$ almost surely.\\
(ii). For the second case we have,
\begin{align*}
  \E\left\{\mathcal W^2(t)-\mathcal W^2(s)\bigg|\mathcal F_s^\mathcal W\right\} &= \E\left\{[\mathcal W(t)-\mathcal W(s)]^2-2\mathcal W(s)[\mathcal W(t)-\mathcal W(s)]\bigg|\mathcal F_s^\mathcal W\right\}\\
  &=\E\left\{[\mathcal W(t)-\mathcal W(s)]^2\bigg|\mathcal F_s^\mathcal W\right\}+2\mathcal W(s)\E\left\{[\mathcal W(t)-\mathcal W(s)]\bigg|\mathcal F_s^\mathcal W\right\}.
\end{align*}
The second part of the above equation vanishes by (i). Furthermore, independence implies
\[
\E\left\{[\mathcal W(t)-\mathcal W(s)]^2\bigg|\mathcal F_s^\mathcal W\right\}=\E\left\{[\mathcal W(t)-\mathcal W(s)]^2\right\}=t-s.
\]
Therefore, $\E\left\{\mathcal W^2(t)-t\bigg|\mathcal F_s^\mathcal W\right\}=\mathcal W^2(s)-s$.\\
(iii). Consider $\Phi$ is standard normal variable with probability density function (pdf) $(2\pi)^{-1/2}\exp(-x^2/2)$. Hence,
\begin{equation*}
  \E\left\{\exp(\gamma \Phi)\right\} =\frac{1}{\sqrt{2\pi}}\int_{-\infty}^\infty \exp(\gamma x)\exp\left\{-\frac{1}{2}x^2\right\}dx=\exp\left\{-\frac{1}{2}\gamma^2\right\},
\end{equation*}
for all $\gamma\in\mathbb R$. Moreover, independence and stationarity implies
\begin{align*}
\E\left\{\exp\left[\sigma\mathcal W(t)-\frac{1}{2}\sigma^2 t\right]\bigg|\mathcal F_s^\mathcal W\right\}&=\exp\left[\sigma\mathcal W(s)-\frac{1}{2}\sigma^2 t\right]\E\left\{\exp\{\sigma[\mathcal W(t)-\mathcal W(s)]\}\bigg|\mathcal F_s^\mathcal W\right\}\\
&=\exp\left[\sigma\mathcal W(s)-\frac{1}{2}\sigma^2 t\right]\E\left\{\exp\{\sigma[\mathcal W(t)-\mathcal W(s)]\}\right\}\\
&=\exp\left[\sigma\mathcal W(s)-\frac{1}{2}\sigma^2 t\right]\E\{\exp\{\sigma\mathcal W(t-s)\}\}.
\end{align*}
Clearly, $\sigma\mathcal W(t-s)\overset{iid}{\sim}\mathcal N\left(0,\sigma^2(t-s)\right)$. It comes from the notion that $\Phi\overset{iid}{\sim}\mathcal N(0,1)$, the random variable $sigma\mathcal W(t-s)$  follows the same law as $\sigma\Phi\sqrt{(t-s)}$, and
\[
\E\{\exp\{\sigma\mathcal W(t-s)\}\}=\E\{\exp\{\sigma\Phi\sqrt{(t-s)}\}\}=\exp\left\{\frac{1}{2}\sigma^2(t-s)\right\}.
\]
Finally,
\[
\E\left\{\exp\left[\sigma\mathcal W(t)-\frac{1}{2}\sigma^2 t\right]\bigg|\mathcal F_s^\mathcal W\right\}=\exp\left\{\sigma\mathcal W(s)-\frac{1}{2}\sigma^2s\right\}
\]
almost surely for every $s<t$. This completes the proof $\square$.

\subsection*{Proof of Lemma \ref{l1}.}

Define $A(\omega):=\inf\{t:\tilde X(t,\omega)\leq -\eta\}$, and $A(t)=A\vee t$. Theorem 6.2.10 of \cite{elliott2005mathematics} implies 
\[
\E\left\{\tilde X[A(t)]\right\}\geq\E\{\tilde X(t)\}.
\]
Hence, 
\begin{equation*}
  \E\{\tilde X(t)\}\leq -\eta P\left[\inf_{s\leq t}\tilde X(s)\leq-\eta\right]+\int_{\left\{\inf_{s\leq t}\tilde X(s)>-\eta\right\}}\tilde X(t)dP,
\end{equation*}
and
\begin{align*}
 \eta P\left[\inf_{s\leq t}\tilde X(s)\leq-\eta\right]&\leq \E\{-\tilde X(t)\} + \int_{\left\{\inf_{s\leq t}\tilde X(s)>-\eta\right\}}\tilde X(t)dP\\
 &=-\int_{\left\{\inf_{s\leq t}\tilde X(s)>-\eta\right\}}\tilde X(t)dP\leq \E\bigg\{\max\left(-\tilde X(t),0\right)\bigg\}.
\end{align*}
Letting $t\ra\infty$ we get our desired result. This completes the proof $\square$.

\subsection*{Proof of Proposition \ref{p2}.}
Define a survival function $F(\hat x^*):=P(\hat X>\hat x^*)$. Now,
\begin{align*}
    \E\left\{\hat X^d\right\}&=-\int_0^\infty(\hat x^*)^d dF(\hat x^*)\\
    &=\int_0^\infty F(\hat x^*)d\left[(\hat x^*)^d\right]-\lim_{k\ra\infty}\left\{(\hat x^*)^d F(\hat x^*)\right\}_0^k\leq\int_0^\infty F(\hat x^*)d\left[(\hat x^*)^d\right]\\
    &\leq \int_0^\infty (\hat x^*)^{-1} \left\{\int_{(\hat X\geq\eta)}\tilde XdP\right\}d\left[(\hat x^*)^d\right]\\
    &=\E\left\{\tilde X\int_0^{\hat X}(\hat x^*)^{-1}d\left[(\hat x^*)^d\right]\right\},\ \ \text{by Fubini's theorem,}\\
    &=\left[\frac{d}{d-1}\right]\E\left\{\tilde X\hat X^{d-1}\right\}\leq\tilde d||\tilde X||_d||\hat X^{d-1}||_{\tilde d},\ \ \text{by H\"older inequality.}
\end{align*}
Therefore, above condition yields
\[
\E\left\{\hat X^d\right\}\leq\tilde d ||\tilde X||_d\left[\E\left\{\hat X^{d\tilde d-\tilde d}\right\}\right]^{1/\tilde d}.
\]
If $\hat X$ is finite the result is trivial because of the fact that $d\tilde d-\tilde d=d$. Contrarily, if $\hat X$ is  not finite, then consider a separate markup such that 
\[
\hat X_k=\hat X\wedge k,\ \ \forall k\in\mathbb N.
\]
This implies $\hat X_k\in L^d$, and $\hat X_k$ satisfies the hypothesis of Proposition \ref{p2}. Thus, $||\hat X||_{\tilde d}\leq\tilde d ||\tilde X||_d$. Allowing $k\ra \infty$ gives the desirable result. This completes the proof $\square$.

\subsection*{Proof of Proposition \ref{p3}.}
If $\eta=\infty$, then the first part of the theorem is trivial because of the fact that $\sup_t||\tilde X(t)||_\infty=K<\infty$, for any finite number $K$. Hence, by above condition $\tilde X(t)\leq K$ almost surely for all $t\in[0,\infty)$. Right-continuity is needed to guarantee that there exists a single set of measure zero outside which this inequality holds for all $t$. Furthermore, for all $\eta\in(1,\infty)$, if $\bar X\in L^d$, then
\[
\sup_t||\tilde X||_d\leq||\bar X||_d<\infty.
\]
It is well understood that $\tilde X(t)$ is uniformly integrable. Thus Corollaries 3.18 and 3.19 of \cite{cohen2015stochastic} yield
\[
\tilde X_\infty(\omega)=\lim_{t\ra\infty}\tilde X(t,\omega),\ \text{a.s.}
\]
Fatou's lemma implies
\[
\E\left\{\lim_t[\tilde X(t)]^d\right\}\leq\lim_t\inf\E\left\{[\tilde X(t)]^d\right\}\leq\sup_t \E\left\{[\tilde X(t)]^d\right\}<\infty.
\]
Hence, $\tilde X_\infty\in L^d$ and $||\tilde X_\infty||_d\leq\sup_t||\tilde X(t)||_d$.

Define $\bar X(t,\omega)=\sup_{s\leq t}\tilde X(s,\omega)$. Thus, $\{-\tilde X(t)\}_{t\in[0,\infty)}$ is a supermartingale. Lemma \ref{l1} implies, for any positive monetary shock
\begin{equation*}
    \eta P\left[\inf_{s\leq t}\left(-\tilde X(s)\right)\leq-\eta\right]=\eta P\left[\bar X(t)\geq\eta\right]\leq \int_{\left[\bar X(t)\geq\eta\right]}\tilde X(t)dP\leq \int_{\left[\bar X\geq\eta\right]}\tilde X(t)dP.
\end{equation*}
Allowing $t\ra\infty$ yields
\[
\eta P\left[\bar X\geq\eta\right]\leq \int_{\left[\bar X\geq\eta\right]}\tilde X_\infty dP.
\]
Finally, Proposition \ref{p2}, with $\hat X=\bar X$, and $\tilde X=\tilde X_\infty$ imply
\[
||\bar X||_d\leq\tilde d ||\tilde X_\infty||_d.
\]
This completes the proof $\square$.

\subsection*{Proof of Proposition \ref{p0}.}
Consider the following Cox-Ingersoll-Ross SDE with a jump
\begin{equation}\label{2}
d\tilde X^{(k)}(s) = \left\{ \tilde\theta \left[u - \tilde X^{(k)}(s) \right] + m^2(s) \right\} ds + \sigma \sqrt{\tilde X^{(k)}(s)} \mathcal W(s)+ \sum_{k:T_k\leq \tilde s}J_{(k)},
\end{equation}
where $\tilde\theta$ is mean reversion rate constant and u is the long term mean of the process. Let
\begin{equation}\label{3}
\tilde X^{(k+1)}(s)=\tilde X(0) + \int_0^t\left\{\tilde\theta \left[u - \tilde X^{(k)} (s)\right] + m^2(s) \right\} ds +\int_0^t \sigma \sqrt{\tilde X^{(k)}(s)} d\mathcal W(s)+ \sum_{k:T_k\leq \tilde s}J_{(k)},
\end{equation}
and
\begin{equation*}
    \E\left\{\left| \tilde X^{(k+1)}(s)-\tilde X^{(k)}(s)\right|^2\right\}\leq((1+z)3S^2\int_0^tE\left\{\left| \tilde X^{(k)}(s)-\tilde X^{(k-1)}(s)\right|^2ds\right\}.
\end{equation*}

Then for $k\geq 1$, $s\leq t$ we have

\begin{align}
  \E\left\{\left|\tilde  X^{(1)}(s)-\tilde X^{(0)}(s)\right|^2\right\}&\leq\{2G^2t^2(1+E|\tilde X(0)|^2)+2G^2t(1+E[|\tilde X(0)|^2]\})\notag\\
  &\leq{N_1t},  
\end{align}

where the constant N(1) only depends on G,S,Z and $E[|\tilde X(0)|^2]$. 

Induction implies

\begin{equation}
 \E\left\{\left|\tilde X^{(k+1)}(s)-\tilde X^{(k)}(s)\right|^2\right\}\leq\frac{N_2^{(k+1)}s^{(k+1)}}{(k+1)!},\ \ \text{$k\geq 0$, and $s\in[0,t]$},   
\end{equation}
for some suitable constant $N{_2}$ depending only on G,S,Z and $E[|\tilde X(0)|^2]$
G is the Lipschitz constant that prevents the process from growing excessively fast and ensures that the process behaves in a controlled manner.
Z is bounded in relation to the stability of the process, ensuring that the system stays within the specified limits throughout its evolution.
S is bounded by the diffusion coefficient  .
Now,

\begin{multline}
\sup_{0\leq s\leq t}\left|\tilde X^{(k+1)}(s)-\tilde X^{(K)}(s)\right|
\leq\int_0^t\left|[\tilde\theta(u-\tilde X^{k}(s))+m^2(s)]-[\tilde\theta(u-\tilde X^{k-1}(s)+m^2(s)]ds\right|\\
+\sup_{0\leq s\leq t}\left|\int_0^t\left(\sigma\sqrt{\tilde X^{(k)}(s)}-\sigma\sqrt{\tilde X^{k-1}(s)}\right)d\mathcal W(s)\right|    \end{multline}

By the Martingale inequality we obtain 
\begin{multline}   
\mathbb{P}\left[\sup_{0 \leq s \leq t} \left|\tilde  X^{(k+1)}(s) -\tilde X^{(k)}(s) \right| > 2^{-k}\right] \\
\leq \mathbb{P}\left[\left( \int_0^t \left| \left[ \tilde\theta(u -\tilde X^{(k)}(s)) + m^2(s) \right] - \left[ \tilde\theta(u - \tilde X^{(k-1)}(s)) + m^2(s) \right] \right| ds \right)^2 > 2^{-2k - 2}\right] \\
+ \mathbb{P}\left[\sup_{0 \leq s \leq t} \left| \int_0^s \left( \sigma \sqrt{\tilde X^{(k)}(s)} - \sigma \sqrt{\tilde X^{(k-1)}(s)} \right) d\mathcal{W}(s) \right| > 2^{-k-1} \right] \\
\leq 2^{2k+2} Z \int_0^t \mathbb{E}\left[ \left| \left[ \tilde\theta(u - \tilde X^{(k)}(s)) + m^2(s) \right] - \left[ \tilde\theta(u - \tilde X^{(k-1)}(s)) + m^2(s) \right] \right|^2 \right] ds \\
+ 2^{-2k-2} S^2(1+Z) \int_0^t \frac{N_2^{k+1} s^k}{k!} dt \\
\leq \frac{(4 N_2)^{k+1} s^{k+1}}{(k+1)!}, \quad \text{if} \quad N{_2} \geq S^2(1+Z)
\end{multline}
Equation \eqref{2} implies

Therefore ,by the Borel-Cantelli lemma
\begin{equation}
 p\left( \sup_{0\leq s\leq t}\left|\tilde X^{(k+1)}(s) - \tilde X^{(k)}(s) \right| > 2^{-k} \text{ for infinitely many } k \right) = 0,
\end{equation}
Thus, for a.s. $\omega$ there exists $k_0 = k_0(\omega)$ such that

\begin{equation}\label{11}
 \sup_{0 \leq s \leq t} \left| \tilde{X}^{(k+1)}(s) - \tilde{X}^{(k)}(s) \right| > 2^{-k},
\end{equation} for $k\geq k_0$
Hence, the sequence

\begin{equation}\label{12}
\tilde{X}^{(n)}(s)(\omega) = \tilde{X}^{(0)}(s)(\omega) + \sum_{k=0}^{n-1} \left( \tilde{X}^{(k+1)}(s)(\omega) - \tilde{X}^{(k)}(s)(\omega) \right)  
\text{is uniformly convergent in } [0, T], \text{ for a.s. } \omega.
\end{equation}
 
 $\omega$ is an element of the sample space which represents all possible outcomes in the probability space.

Let the limit be denoted by \( Y_t = Y_t(\omega) \). Then \( Y_t \) is \( s \) continuous for almost all \( \omega \) since \( Y^{(n)} \) is \( s \)-continuous for all \( n \). Moreover, \( Y(\cdot) \) is \( \mathcal{F}^Z(s) \)-measurable for all \( s \), since \( Y^{(n)}(\cdot) \) has this property for all \( n \).
Furthermore , observe that for \( m > n \geq 0 \), we have by (7)

\begin{multline}\label{13}
\mathbb{E}\left[ |\tilde{X}^{(m)}(s) - \tilde{X}^{(n)}(s)|^2 \right]^{1/2} = \|\tilde{X}^{(m)}(s) - \tilde{X}^{(n)}(s)\|_{L^2(P)} \\
= \left\| \sum_{k=n}^{m-1} \left( \tilde{X}^{(k+1)}(s) - \tilde{X}^{(k)}(s) \right) \right\|_{L^2(P)} 
\leq \sum_{k=n}^{m-1} \left\|\tilde{X}^{(k+1)}(s) - \tilde{X}^{(k)}(s) \right\|_{L^2(P)} \\
\leq \| \sum_{k=n}^{\infty} \left[ \frac{(N_2t)^{(k+1)}}{(k+1)!}\right]^{1/2} \to 0 \quad \text{as} \quad n \to \infty.
\end{multline}
The sequence \( \{\tilde X^{(n)}(s)\} \) converges to a limit, denoted \(\tilde X(s) \), in \( L^2(P) \). A subsequence of \(\tilde X^{(n)}(\omega) \) will then converge in \( \omega \) pointwise to \(\tilde X(s)(\omega) \) for almost every \( \omega \), which implies that \( \tilde X(s) = Y(s) \) almost surely.

Now we show that $Y(s)$ satisfies (2).
For all n we have ,
\begin{equation}\label{14}
 \tilde {X}^{n+1}(s) = Y(0) + \int_0^t \tilde{\theta} \left[ (u - \tilde{X}^{n}(s)) + m^2(s) \right] ds + \int_0^t \sigma \sqrt{\tilde{X}^{n}(s)} dW_s
\end{equation}

Now, $ \tilde {X}^{n+1}(s)\to Y(s) \quad \text{as} \quad n \to \infty, \quad \text{uniformly in} \quad t \in [0,T] \quad \text{for a.s.} \ \omega.$. By (13) and the Fatou lemma, we have
\begin{equation}\label{15}
    \mathbb{E} \left[ \int_0^t |Y(s) - \tilde X^{n}(s)|^2 \, dt \right] \leq \limsup_{m \to \infty} \mathbb{E} \left[ \int_0^t |\tilde X^{m}(s) -\tilde X^{n}(s)|^2 \, dt \right] \to 0
\end{equation}
$\text{as } n \to \infty$.

$\text{It follows by the It\^{o} isometry that} \quad \int_0^t \sigma \sqrt{\tilde{X}^{n}(s)} \, dW_s \to \int_0^t \sigma \sqrt{\tilde{X}^{n}(s)} \, dW_s.$  Finally, by the H\"older inequality that $ \int_0^t \tilde{\theta} \left[ (u - \tilde{X}^{n}(s)) + m^2(s) \right] \, ds \to \int_0^t \tilde{\theta} \left[ (u - Y^{n}(s)) + m^2(s) \right] \, ds$ in $L^2(p)$ .
Therefore, taking the limit of (14) as \( n \to \infty \), we obtain \begin{equation}\label{2}
d\tilde X^{(k)}(s) = \left\{ \tilde\theta \left[u - \tilde X^{(k)}(s) \right] + m^2(s) \right\} ds + \sigma \sqrt{\tilde X^{(k)}(s)} \mathcal W(s)+ \sum_{k:T_k\leq \tilde s}J_{(k)}.
\end{equation}
This completes the proof. $\square$

\subsection*{Proof of Proposition \ref{p4}.}

We consider the Cox-Ingersoll-Ross (CIR) stochastic differential equation:  
\[
d\tilde{X}(s) = \bigg\{\tilde{\theta} \left[u - \tilde{X}(s)\right] + m^2(s)\bigg\} ds + \sigma \sqrt{\tilde{X}(s)} d\mathcal{W}(s) + \sum_{k:T_k \leq s} J_{(k)}.
\]  
Here, \( \tilde{\theta} > 0 \) is the mean-reversion rate, \( u > 0 \) is the long-term mean, \( m(s) \) accounts for strategic complementarity, \( \sigma > 0 \) is the diffusion coefficient, and \( J_{(k)} \) represents jumps at times \( T_k \). The goal is to solve this SDE using an integrating factor that depends on both the state variable \( \tilde{X}(s) \) and time \( s \).  

The given equation can be expressed equivalently as:  
\[
d\tilde{X}(s) = \tilde{\theta} u \, ds - \tilde{\theta} \tilde{X}(s) \, ds + m^2(s) \, ds + \sigma \sqrt{\tilde{X}(s)} \, d\mathcal{W}(s) + \sum_{k:T_k \leq s} J_{(k)}.
\]  
Rearranging, we write:  
\[
d\tilde{X}(s) + \tilde{\theta} \tilde{X}(s) \, ds = \left(\tilde{\theta} u + m^2(s)\right) ds + \sigma \sqrt{\tilde{X}(s)} \, d\mathcal{W}(s) + \sum_{k:T_k \leq s} J_{(k)}.
\]

To simplify the solution, we introduce the integrating factor:  
\[
I(s, \tilde{X}(s)) = e^{\tilde{\theta} s} \cdot \frac{1}{\tilde{X}(s)}.
\]  
Multiplying through by \( I(s, \tilde{X}(s)) \), the equation becomes:  
\[
e^{\tilde{\theta} s} \frac{1}{\tilde{X}(s)} d\tilde{X}(s) + \tilde{\theta} e^{\tilde{\theta} s} \, ds = e^{\tilde{\theta} s} \frac{\tilde{\theta} u}{\tilde{X}(s)} \, ds + e^{\tilde{\theta} s} \frac{m^2(s)}{\tilde{X}(s)} \, ds + e^{\tilde{\theta} s} \frac{\sigma}{\sqrt{\tilde{X}(s)}} \, d\mathcal{W}(s) + e^{\tilde{\theta} s} \frac{1}{\tilde{X}(s)} \sum_{k:T_k \leq s} J_{(k)}.
\]  

Using the fact that \( \frac{1}{\tilde{X}(s)} d\tilde{X}(s) = d\ln(\tilde{X}(s)) \), the left-hand side simplifies to:  
\[
e^{\tilde{\theta} s} d\ln(\tilde{X}(s)) + \tilde{\theta} e^{\tilde{\theta} s} \, ds = \frac{d}{ds} \left(e^{\tilde{\theta} s} \ln(\tilde{X}(s))\right).
\]  
Thus, the equation reduces to:  
\[
\frac{d}{ds} \left(e^{\tilde{\theta} s} \ln(\tilde{X}(s))\right) = e^{\tilde{\theta} s} \bigg\{\frac{\tilde{\theta} u}{\tilde{X}(s)} + \frac{m^2(s)}{\tilde{X}(s)}\bigg\} + e^{\tilde{\theta} s} \frac{\sigma}{\sqrt{\tilde{X}(s)}} \, d\mathcal{W}(s) + e^{\tilde{\theta} s} \frac{1}{\tilde{X}(s)} \sum_{k:T_k \leq s} J_{(k)}.
\]

Integrating both sides with respect to \( s \) over the interval \([s_0, s]\), we obtain:  
\begin{align*}
    e^{\tilde{\theta} s} \ln(\tilde{X}(s)) - e^{\tilde{\theta} s_0} \ln(\tilde{X}(s_0)) &= \int_{s_0}^s e^{\tilde{\theta} \xi} \bigg\{\frac{\tilde{\theta} u}{\tilde{X}(\xi)} + \frac{m^2(\xi)}{\tilde{X}(\xi)}\bigg\} d\xi\\
    &\hspace{1cm}+ \int_{s_0}^s e^{\tilde{\theta} \xi} \frac{\sigma}{\sqrt{\tilde{X}(\xi)}} \, d\mathcal{W}(\xi) + \int_{s_0}^s e^{\tilde{\theta} \xi} \frac{1}{\tilde{X}(\xi)} \sum_{k:T_k \leq \xi} J_{(k)} d\xi.
\end{align*}

Rearranging terms and taking the exponential of both sides yields the explicit solution for \( \tilde{X}(s) \):  
\begin{align*}
      \tilde{X}^*(s) &= \exp\bigg\{\exp\left\{-\tilde{\theta} s\right\} \bigg[\exp\left\{\tilde{\theta} s_0\right\} \ln(\tilde{X}(s_0)) + \int_{s_0}^s \exp\left\{\tilde{\theta} \xi\right\} \bigg\{\frac{\tilde{\theta} u}{\tilde{X}(\xi)} + \frac{m^2(\xi)}{\tilde{X}(\xi)}\bigg\} d\xi\\
      &\hspace{1cm}+ \int_{s_0}^s \exp\left\{\tilde{\theta} \xi\right\} \frac{\sigma}{\sqrt{\tilde{X}(\xi)}} \, d\mathcal{W}(\xi) + \int_{s_0}^s \exp\left\{\tilde{\theta} \xi\right\} \frac{1}{\tilde{X}(\xi)} \sum_{k:T_k \leq \xi} J_{(k)} d\xi\bigg]\bigg\},
\end{align*}

This provides the general form of the solution, where the specific evaluation depends on the structure of \( m^2(s) \), \( J_{(k)} \), and the initial condition \( \tilde{X}(s_0) \). This completes the proof. $\square$

\subsection*{Proof of Lemma \ref{j0}.}
The jump diffusion to \( \tilde{X}(t) \) is $
\sum_{k:T_k \leq t} J_{(k)},$
where \( T_k \) denotes the times of jumps and \( J_{(k)} \) are the corresponding jump sizes. Since the jumps are driven by a Poisson process with intensity \( \nu \), the number of jumps \( N(t) \) over the interval \( [0,t] \) follows a Poisson distribution with mean \( \mathbb{E}[N(t)] = \nu t \).

Taking the expectation yields
\[
\mathbb{E} \left[ \sum_{k:T_k \leq t} J_{(k)} \right] = \mathbb{E} \left[ N(t) \cdot \mathbb{E}[J_{(k)}] \right].
\]
Since \( J_{(k)} \) are assumed to be independent and identically distributed  with mean \( \gamma \) and independent of \( N(t) \), we have,
\[
\mathbb{E} \left[ \sum_{k:T_k \leq t} J_{(k)} \right] = \mathbb{E}[N(t)] \cdot \gamma.
\]
Substituting \( \mathbb{E}[N(t)] = \nu t \), yields
\[
\mathbb{E} \left[ \sum_{k:T_k \leq t} J_{(k)} \right] = \nu t \gamma.
\]
This completes the proof. $\square$

\subsection*{Proof of Lemma \ref{j1}.}

The variance of this sum can be computed using the law of total variance,
\[
\text{Var} \left( \sum_{k:T_k \leq t} J_{(k)} \right) = \mathbb{E}[N(t)] \cdot \text{Var}(J_{(k)}) + \text{Var}[N(t)] \cdot (\mathbb{E}[J_{(k)}])^2.
\]
For the Poisson process, \( N(t)\overset{iid}{\sim} \text{Poisson}(\nu t) \), so \( \mathbb{E}[N(t)] = \nu t \) and \( \text{Var}[N(t)] = \nu t \). Substituting these into the equation
\[
\text{Var} \left( \sum_{k:T_k \leq t} J_{(k)} \right) = (\nu t) \cdot \sigma_J^2 + (\nu t)  \gamma^2.
\]
Rearranging the terms yields
\[
\text{Var} \left( \sum_{k:T_k \leq t} J_{(k)} \right) = \nu t (\gamma^2 + \sigma_J^2).
\]
This complete the proof. $\square$

\subsection*{Proof of Proposition \ref{p5}.}
We begin by considering the state variable \( \tilde{X}(s) \), which follows a jump-diffusion model governed by the following SDE
\[
d\tilde{X}(s) = \left[ \tilde{\theta} (u - \tilde{X}(s)) + m^2(s) \right] ds + \sigma \sqrt{\tilde{X}(s)} dW(s) + \sum_{k: T_k \leq s} J_{(k)},
\]
where \( \tilde{\theta} \) is the rate at which the process reverts to the mean, \( u \), and \( m(s) \) represents the effect of a monetary shock. The term \( \sum_{k: T_k \leq s} J_{(k)} \) describes the jumps at specific times \( T_k \), and \( J_{(k)} \) represents the size of the jump at time \( T_k \). The jump sizes \( J_{(k)} \) are modeled using a Lévy process with independent increments, and the Poisson process governing the jumps has intensity \( \nu \). The CIR process without jumps exhibits a stationary behavior, with a long-term mean \( u \) and variance \( \frac{\sigma^2}{2 \tilde{\theta}} \). As time progresses, the state variable \( \tilde{X}(t) \) will tend toward this stationary distribution, given that the process remains stable.

Next, we analyze the jump component, \( \sum_{k: T_k \leq t} J_{(k)} \), which is governed by a Poisson process with intensity \( \nu \). The expected total number of jumps at time \( t \) grows linearly with time, specifically, \( \mathbb{E}[N(t)] = \nu t \). Since the jumps are independent, with each jump \( J_{(k)} \) having a mean of \( \lambda \) and a variance of \( \sigma_J^2 \), the expected total jump contribution at time \( t \) is
\[
\mathbb{E}\left[\sum_{k: T_k \leq t} J_{(k)}\right] = \nu t \gamma.
\]

Similarly, the variance of the total jump diffusion is
\[
\text{Var}\left[\sum_{k: T_k \leq t} J_{(k)}\right] = \nu t  (\gamma^2 + \sigma_J^2).
\]

As time progresses, the number of jumps increases proportionally with \( t \), and consequently, the total jump contribution also increases with \( t \), with both the expected value and the variance scaling linearly with \( t \). To understand the joint behavior of \( \tilde{X}(t) \) and \( \sum_{k: T_k \leq t} J_{(k)} \), we note that the diffusion component of \( \tilde{X}(t) \) tends to a stationary distribution with a mean of \( u \) and a variance of \( \frac{\sigma^2}{2 \tilde{\theta}} \). Since the jumps are independent of the continuous part of the process, the covariance between \( \tilde{X}(t) \) and the jump process is zero. Therefore, the joint distribution of \( \tilde{X}(t) \) and \( \sum_{k: T_k \leq t} J_{(k)} \) as \( t \to \infty \) will approach a bivariate normal distribution with the following properties:

\[
\mathbb{E}[\tilde{X}(t)] = u, \quad \mathbb{E}\left[\sum_{k: T_k \leq t} J_{(k)}\right] = \nu t \gamma,
\]
\[
\text{Var}(\tilde{X}(t)) = \frac{\sigma^2}{2 \tilde{\theta}}, \quad \text{Var}\left(\sum_{k: T_k \leq t} J_{(k)}\right) = \nu t (\gamma^2 + \sigma_J^2),
\]
and the covariance between \( \tilde{X}(t) \) and \( \sum_{k: T_k \leq t} J_{(k)} \) is zero.

Since time approaches infinity, the joint distribution of \( \tilde{X}(t) \) and \( \sum_{k: T_k \leq t} J_{(k)} \) converges to a bivariate Gaussian distribution. This implies that, in the long term, the system's state variable and the cumulative effect of jumps both follow normal distributions, and their joint distribution is independent of the diffusion process. This completes the proof. $\square$

\subsection*{Proof of Proposition \ref{p6}.}

Define a Lyapunov function $V(\tilde{X}(s)) := (\tilde{X}(s) - u)^2$, such that
\[
V(\tilde{X}(s)) =\begin{cases}
    0, \ \text{if}\ \tilde{X}(s) = u,\\
    >0, \ \text{for all}\ \tilde{X}(s) \neq u,\\
    \infty,\ \text{if}\ \tilde{X}(s) \to \infty.
\end{cases}
\]

Therefore, \( V(\tilde{X}(s)) \) is positive definite and radially unbounded, satisfying the standard conditions for Lyapunov stability analysis.

Applying It\^o Lemma to \( V(\tilde{X}(s)) = (\tilde{X}(s) - u)^2 \) yields
\begin{equation*}
    dV(\tilde{X}(s)) = 2 (\tilde{X}(s) - u) d\tilde{X}(s) + \sigma^2 ds.
\end{equation*}

Substituting the dynamics of \( \tilde{X}(s) \) yields

\[
dV(\tilde{X}(s)) = 2 (\tilde{X}(s) - u) \left( \tilde{\theta} (u - \tilde{X}(s)) + m^2(s) \right) ds + 2 (\tilde{X}(s) - u) \sigma \sqrt{\tilde{X}(t)} d\mathcal W(s) + 2 (\tilde{X}(s) - u) \sum_{k:T_k \leq t} J(k).
\]

Taking the expectation of both sides, and noting that the expectations of the stochastic terms \( d\mathcal W(s) \) and \( \sum_{k:T_k \leq t} J(k) \) are zero (since \( \E[d\mathcal W(s)] = 0 \) and \( \E[J(k)] = \gamma \), the mean jump size), 

\[
\E[dV(\tilde{X}(s))] = \E\left[ 2 (\tilde{X}(t) - u) \left( \tilde{\theta} (u - \tilde{X}(s)) + m^2(s) \right) \right].
\]

We now analyze the expression inside the expectation. The first term, \( 2 (\tilde{X}(s) - u) \tilde{\theta} (u - \tilde{X}(s)) \), is of the form:

\[
2 (\tilde{X}(s) - u) \tilde{\theta} (u - \tilde{X}(s)) = -2 \tilde{\theta} (\tilde{X}(s) - u)^2.
\]

This term is always negative for \( \tilde{X}(s) \neq u \), since \( \tilde{\theta} > 0 \) and the square of any non-zero quantity is positive. This term provides a restoring force that drives \( \tilde{X}(s) \) back toward \( u \). The second term, \( 2 (\tilde{X}(s) - u) m^2(t) \), represents the contribution of the control variable. Since \( m^2(s) \) is assumed to be bounded (i.e., \( |m^2(t)| \leq M \), where \( M \) is a constant), this term does not destabilize the process but rather introduces a bounded fluctuation around the equilibrium. Finally, the jump term \( 2 (\tilde{X}(t) - u) \sum_{k:T_k \leq t} J(k) \) has zero mean, as \( \E[J(k)] = \gamma \), the expected jump size, and \( \tilde{X}(s) - u \) is centered around \( u \). Therefore, 

\[
\E\left[ 2 (\tilde{X}(t) - u) \sum_{k:T_k \leq t} J(k) \right] = 0.
\]

Thus, the expected time derivative of the Lyapunov function is

\[
\E[dV(\tilde{X}(s))] = -2 \tilde{\theta} (\tilde{X}(s) - u)^2 + O(m^2(s)).
\]

For Lyapunov stability, we require that \( \E[dV(\tilde{X}(s))] \) is negative definite. This is ensured if

\[
-2 \tilde{\theta} (\tilde{X}(s) - u)^2 + O(m^2(s)) < 0 \quad \text{for all} \quad \tilde{X}(s) \neq u.
\]

Since \( \tilde{\theta} > 0 \) and \( m^2(s) \) is bounded, the drift term ensures that \( \tilde{X}(s) \) is driven towards \( u \), and the fluctuation from the jump diffusion does not prevent convergence to the equilibrium. Specifically, the presence of random shocks (modeled by \( m^2(s) \) and the jumps) does not destabilize the system as long as the intensity of the jumps is sufficiently large and the shocks remain bounded. This completes the proof. $\square$

\subsection*{Proof of Theorem \ref{t0}.}
The Euclidean action function of an atomistic firm is
\begin{align}
	\mathcal A_{0,t}(\tilde X)&=\int_0^t\E_s\bigg\{\exp(-\rho s)\Theta[s,m(s),\tilde X]ds\notag\\
	&\hspace{1cm}+\bigg[\tilde x(s)-\tilde x_0-\mu(s,m,\tilde X)ds-\sigma(s,m,\tilde X)dB(s)\bigg]d\lambda(s)\bigg\},\notag
	\end{align}
	where $E_s$ is the conditional expectation on markup dynamics $\tilde X(s)$ at the beginning of time $s$. For all $\varepsilon>0$, and~the penalizing constant $L_\varepsilon>0$ , define a transitional probability in infinitesimal time interval as
\begin{align}\label{w16}
	\Psi_{s,s+\varepsilon}(\tilde X)&:=\frac{1}{L_\varepsilon} \int_{\mathbb{R}} \exp\biggr\{-\varepsilon  \mathcal{A}_{s,s+\varepsilon}(\tilde X)\biggr\} \Psi_s(\tilde X) d\tilde X(s),
	\end{align}	
	for $\epsilon\downarrow 0$, and $\Psi_s(\tilde X)$ is the value of the transition probability at $s$ and markup dynamics $\tilde X(s)$ with initial condition $\Psi_0(\tilde X)=\Psi_0$.

    For continuous time interval $[s,\tau]$, such as $\tau=s+\varepsilon$,  the stochastic Lagrangian is defines as
\begin{align}\label{action}
	\mathcal{A}_{s,\tau}(\tilde X)&= \int_{s}^{\tau} \E_s\biggr\{\exp(-\rho \nu)\Theta\left[\nu,m(\nu),\tilde X(\nu)\right] d\nu\notag\\
	&\hspace{1cm}+\bigg[\tilde x(\nu)-\tilde x_0-\mu\left[\nu,m(\nu),\tilde X(\nu)\right]d\nu-\sigma\left[\nu,m(\nu),\tilde X(\nu)\right]dB(\nu)\bigg]d\lambda(\nu)\bigg\},
	\end{align}
	with the constant initial condition $\tilde x(0)=\tilde x_0$.	This conditional expectation holds when the function \( m(\nu) \) governing a firm's markup dynamics is established at time \( \nu \), assuming the markup dynamics of all other firms are predetermined. The evolution proceeds as the action function remains stationary. Consequently, the conditional expectation over time is solely influenced by the expectation at the initial time point of the interval \([s, \tau]\).

    By Fubini's Theorem,
\begin{align}\label{action5}
	\mathcal{A}_{s,\tau}(\tilde X)&= \E_s\ \bigg\{ \int_{s}^{\tau}\exp(-\rho\nu)\Theta\left[\nu,m(\nu),\tilde X(\nu)\right] d\nu\notag\\
	&\hspace{1cm}+\bigg[\tilde x(\nu)-\tilde x_0-\mu\left[\nu,m(\nu),\tilde X(\nu)\right]d\nu-\sigma\left[\nu,m(\nu),\tilde X(\nu)\right]dB(\nu)\bigg]d\lambda(\nu) \bigg\}.
	\end{align}
	By It\^o's Theorem, there exists a function $\tilde h[\nu,\tilde X(\nu)]\in C^2([0,\infty)\times\mathbb{R})$ such that  $\Upsilon(\nu)=\tilde h[\nu,\tilde X(\nu)]$, where $\Upsilon(\nu)$ is an It\^o~process. Assuming 
	\[
	\tilde h[\nu+\Delta \nu,\tilde X(\nu)+\Delta \tilde X(\nu)]= \tilde x(\nu)-\tilde x_0-\mu\left[\nu,m(\nu),\tilde X(\nu)\right]d\nu-\sigma\left[\nu,m(\nu),\tilde X(\nu)\right]dB(\nu),
	\]
 Equation~(\ref{action5}) yields,
\begin{align}\label{action6}
	\mathcal{A}_{s,\tau}(\tilde X)&=\E_s \bigg\{ \int_{s}^{\tau}\exp(-\rho\nu)\Theta\left[\nu,m(\nu),\tilde X(\nu)\right] d\nu+ \tilde h\left[\nu+\Delta \nu,\tilde X(\nu)+\Delta \tilde X(\nu)\right]d\lambda(\nu)\bigg\}.
	\end{align}

    By It\^o's Lemma,
\begin{align}\label{action7}
	\varepsilon\mathcal{A}_{s,\tau}(\tilde X)&= \E_s \bigg\{\varepsilon \exp(-\rho s)\Theta\left[\nu,m(\nu),\tilde X(\nu)\right]+ \varepsilon \tilde h[s,\tilde X(s)]d\lambda(s)+ \varepsilon \tilde h_s[s,\tilde X(s)]d\lambda(s) \notag\\
	&\hspace{.25cm}+\varepsilon \tilde h_{\tilde X}[s,\tilde X(s)]\mu\left[\nu,m(\nu),\tilde X(\nu)\right]d\lambda(s) \notag\\
    &\hspace{.5cm}+\varepsilon \tilde h_{\tilde X}[s,\tilde X(s)]\sigma\left[\nu,m(\nu),\tilde X(\nu)\right]d\lambda(s) dB(s)\notag\\
	&\hspace{1cm}+\mbox{$\frac{1}{2}$}\varepsilon\left(\sigma\left[s,x(s),u(s)\right]\right)^2h_{\tilde X\tilde X}[s,\tilde X(s)]d\lambda(s)+o(\varepsilon)\bigg\},
	\end{align}
	where $\tilde h_s=\frac{\partial}{\partial s} h$, $h_{\tilde X}=\frac{\partial}{\partial \tilde X} h$ and $h_{\tilde X\tilde X}=\frac{\partial^2}{\partial (\tilde X)^2} h$, and~we use the condition $[d \tilde X(s)]^2\approx\varepsilon$ with 
    \[
d\tilde X(s)\approx\varepsilon\mu\left[s,m(s),\tilde X(s)\right]+\sigma\left[s,m(s),\tilde X(s)\right]dB(s).
\]
We apply It\^o's Lemma along with a comparable approximation to estimate the integral. As \(\varepsilon\) approaches zero, dividing by \(\varepsilon\) and taking the conditional expectation results in
\begin{align}\label{action8}
	\varepsilon\mathcal{A}_{s,\tau}(\tilde X)&= \E_s \bigg\{\varepsilon \exp(-\rho s)\Theta\left[s,m(s),\tilde X(s)\right]+\varepsilon \tilde h[s,\tilde X(s)]d\lambda(s)+ \varepsilon \tilde h_s[s,\tilde X(s)]d\lambda(s)\notag\\
	&\hspace{.25cm}+\varepsilon \tilde h_{\tilde X}[s,\tilde X(s)]\mu\left[s,m(s),\tilde X(s)\right]d\lambda(s)\notag\\
	&\hspace{.5cm}+\mbox{$\frac{1}{2}$}\varepsilon\sigma^{2}\left[s,m(s),\tilde X(s)\right]h_{\tilde X\tilde X}[s,\tilde X(s)]d\lambda(s)+o(1)\bigg\},
	\end{align}
	since $\E_s[dB(s)]=0$ and $\E_s[o(\varepsilon)]/\varepsilon\ra 0$ for all $\varepsilon\downarrow 0$. For~$\varepsilon\downarrow 0$, denote a transition probability at $s$ as $\Psi_s(\tilde X)$. By Equation~(\ref{w16}),
	\begin{multline}\label{action9}
	\Psi_{s,\tau}(\tilde X)=\frac{1}{L_\epsilon^i}\int_{\mathbb{R}} \exp\biggr\{-\varepsilon \big[\exp(-\rho s)\Theta\left[s,m(s),\tilde X(s)\right]+\tilde h[s,\tilde X(s)]d\lambda(s)\\
	 +\tilde h_s[s,\tilde X(s)]d\lambda(s) +h_{\tilde X}[s,\tilde X(s)]\mu\left[s,m(s),\tilde X(s)\right]d\lambda(s)\\
	+\mbox{$\frac{1}{2}$}\left(\sigma\left[s,m(s),\tilde X(s)\right]\right)^2\tilde h_{\tilde X\tilde X}[s,\tilde X(s)]d\lambda(s)\big]\biggr\} \Psi_s(\tilde X) d\tilde X(s)+o(\varepsilon^{1/2}).
	\end{multline}

    Since $\varepsilon\downarrow 0$, first-order Taylor series expansion on the left-hand side of Equation~(\ref{action9}) gives
	\begin{multline}\label{action10}
	\Psi_{s}(\tilde X)+\varepsilon  \frac{\partial \Psi_{s}(\tilde X) }{\partial s}+o(\varepsilon)=\frac{1}{L_\varepsilon}\int_{\mathbb{R}} \exp\biggr\{-\varepsilon \big[\exp(-\rho s)\Theta\left[s,m(s),\tilde X(s)\right]+\tilde h[s,\tilde X(s)]d\lambda(s) \\
+\tilde h_s[s,\tilde X(s)]d\lambda(s)+\tilde h_{\tilde X}[s,\tilde X(s)]\mu\left[s,m(s),\tilde X(s)\right]d\lambda(s)\\
	+\mbox{$\frac{1}{2}$}\left(\sigma\left[s,m(s),\tilde X(s)\right]\right)^2\tilde h_{\tilde X\tilde X}[s,\tilde X(s)]d\lambda(s)\big]\biggr\} \Psi_s(\tilde X) d\tilde X(s)+o(\varepsilon^{1/2}).
	\end{multline}
    Now define $\tilde X(s)-\tilde X(\tau):=\tilde \xi$ such that $\tilde X(s)=\tilde X(\tau)+\tilde\xi$. For $\tilde\xi\neq 0$, and for~a positive number $\eta<\infty$ assume $|\tilde\xi|\leq\sqrt{\frac{\eta\varepsilon}{\tilde X(s)}}$ such that for $\varepsilon\downarrow 0$, $\tilde\xi$ attains smaller values and the markup dynamics $0<\tilde X(s)\leq\eta\varepsilon/(\tilde\xi)^2$. Therefore,
	\begin{multline*}
	\Psi_{s}(\tilde X)+\varepsilon\frac{\partial \Psi_{s}(\tilde X)}{\partial s}=\frac{1}{L_\epsilon}\int_{\mathbb{R}} \left[\Psi_{s}(\tilde X)+\tilde\xi\frac{\partial \Psi_{is}(\tilde X)}{\partial \tilde X}+o(\epsilon)\right]\\
	\times \exp\biggr\{-\varepsilon \big[\exp(-\rho s)\Theta\left[s,m(s),\tilde X\right]+\tilde h[s,\tilde X(s)]d\lambda(s)\\
    +h_{\tilde X}[s,\tilde X(s)]\mu\left[s,m(s),\tilde X(s)\right]d\lambda(s)\\
	+\mbox{$\frac{1}{2}$}\left(\sigma\left[s,m(s),\tilde X(s)\right]\right)^2h_{\tilde X\tilde X(s)}[s,\tilde X(s)]d\lambda(s)\big]\biggr\} d\tilde\xi+o(\varepsilon^{1/2}).
	\end{multline*}
    Before computing Gaussian integral of  each term of the right-hand side of the above Equation, define a $C^2$ function 
	\begin{align*}
	\ell[s,\tilde\xi,\lambda(s),m(s)]&=\exp(-\rho s)\Theta\left[s,m(s),\tilde X(s)+\tilde\xi\right]+\tilde h\left[s,\tilde X(s)+\tilde\xi\right]d\lambda(s) +\tilde h_s\left[s,\tilde X(s)+\tilde\xi\right]d\lambda(s)\notag\\
&\hspace{.25cm}+\tilde h_{\tilde X}\left[s,\tilde X(s)+\tilde\xi\right]\mu\left[s,m(s),\tilde X(s)+\tilde\xi\right]d\lambda(s)\\
	&\hspace{1cm}+\mbox{$\frac{1}{2}$}\sigma^{2}\left[s,m(s),\tilde X(s)+\tilde\xi\right]h_{\tilde X\tilde X}\left[s,\tilde X(s)+\tilde\xi\right]d\lambda(s)+o(1).
	\end{align*}
     Therefore,
\begin{align}\label{action13}
	\Psi_{s}(\tilde X)+\varepsilon \frac{\partial \Psi_{s}(\tilde X) }{\partial s}&=\Psi_{s}(\tilde X)\frac{1}{L_\epsilon}\int_{\mathbb{R}}\exp\left\{-\varepsilon \ell\left[s,\tilde\xi,\lambda(s),m(s)\right] \right\}d\tilde\xi\notag\\
	&+\frac{\partial \Psi_{s}(\tilde X)}{\partial \tilde X}\frac{1}{L_\epsilon}\int_{\mathbb{R}}\tilde\xi\exp\left\{-\varepsilon \ell\left[s,\tilde\xi,\lambda(s),m(s)\right] \right\}d\tilde\xi+o(\varepsilon^{1/2}).
	\end{align}
After taking $\varepsilon\downarrow 0$, $\Delta m\downarrow0$ and a Taylor series expansion with respect to $\tilde X$ of $\ell\left[s,\tilde\xi,\lambda(s),m(s)\right]$ yields, 
	\begin{align*}
	\ell\left[s,\tilde\xi,\lambda(s),m(s)\right]&=\ell[s, \tilde X(\tau),\lambda(s),m(s)]+\ell_{\tilde X}\left[s, \tilde X(\tau),\lambda(s),m(s)\right]\left[\tilde\xi-\tilde X(\tau)\right]\notag\\
	&\hspace{1cm}+\mbox{$\frac{1}{2}$}\ell_{\tilde X\tilde X}\left[s, \tilde X(\tau),\lambda(s),m(s)\right]\left[\tilde\xi-\tilde X(\tau)\right]^2+o(\varepsilon).
	\end{align*}
	Define $\tilde Y:=\tilde \xi-\tilde X(\tau)$ so that $ d\tilde \xi=d\tilde Y$. The first integral on the right-hand side of Equation~(\ref{action13}) yields
\begin{align}\label{action14}
	&\int_{\mathbb{R}} \exp\big\{-\varepsilon \ell\left[s,\tilde\xi,\lambda(s),m(s)\right]\} d\tilde\xi\notag\\
	&=\exp\big\{-\varepsilon \ell[s, \tilde X(\tau),\lambda(s),m(s)]\big\}\notag\\
	&\hspace{1cm}\int_{\mathbb{R}} \exp\biggr\{-\varepsilon \biggr[\ell_{\tilde X}[s, \tilde X(\tau),\lambda(s),m(s)]y+\mbox{$\frac{1}{2}$}\ell_{\tilde X\tilde X}[s,\tilde X(\tau),\lambda(s),m(s)] \tilde Y^2\biggr]\biggr\} d\tilde Y.
	\end{align}
Assuming  $\tilde a=\frac{1}{2} \ell_{\tilde X\tilde X}[s, \tilde X(\tau),\lambda(s),m(s)]$ and $\tilde b=\ell_{\tilde X}[s, \tilde X(\tau),\lambda(s),m(s)]$ the argument of the exponential function in Equation~(\ref{action14}) becomes,
\begin{align}\label{action15}
	\tilde a\tilde Y^2+\tilde b\tilde Y&=\tilde a\left[(\tilde Y)^2+\frac{\tilde b}{\tilde a}\tilde Y\right]\approx \tilde a\left(\tilde Y+\frac{\tilde b}{2\tilde a}\right)^2-\frac{(\tilde b)^2}{4(\tilde a)^2},
	\end{align}
    as $\tilde a>0$ and $\tilde a\ra 0$.

    Therefore,
\begin{align}\label{action16}
	&\exp\bigg\{-\varepsilon \ell[s, \tilde X(\tau),\lambda(s),m(s)]\bigg\}\int_{\mathbb{R}} \exp\big\{-\varepsilon [\tilde a(\tilde Y)^2+\tilde b\tilde Y]\big\}d\tilde Y\notag\\
	&=\exp\left\{\varepsilon \left[\frac{\tilde b^2}{4\tilde a^2}-\ell[s, \tilde X(\tau),\lambda(s),m(s)]\right]\right\}\int_{\mathbb{R}} \exp\left\{-\left[\varepsilon \tilde a\left(\tilde Y+\frac{\tilde b}{2\tilde a}\right)^2\right]\right\} d\tilde Y\notag\\
	&=\sqrt{\frac{\pi}{\varepsilon \tilde a}}\exp\left\{\varepsilon \left[\frac{\tilde b^2}{4\tilde a^2}-\ell[s, \tilde X(\tau),\lambda(s),m(s)]\right]\right\},
	\end{align}
	and
\begin{align}\label{action17}
	&\Psi_{s}(\tilde X) \frac{1}{L_\varepsilon} \int_{\mathbb{R}} \exp\big\{-\varepsilon \ell\left[s,\tilde\xi,\lambda(s),m(s)\right]\} d\tilde\xi\notag\\
	 &=\Psi_{s}(\tilde X) \frac{1}{L_\varepsilon} \sqrt{\frac{\pi}{\varepsilon \tilde a}}\exp\left\{\varepsilon \left[\frac{\tilde b^2}{4\tilde a^2}-\ell[s, \tilde X(\tau),\lambda(s),m(s)]\right]\right\}. 
	\end{align}

Substituting $\tilde\xi=\tilde X(\tau)+\tilde Y$ into the second integrand of the right-hand side of Equation~(\ref{action13}) yields
\begin{align}\label{action18}
	& \int_{\mathbb{R}} \tilde\xi \exp\left[-\varepsilon\left \{\ell\left[s,\tilde\xi,\lambda(s),m(s)\right]\right\}\right] d\tilde\xi\notag\\
	&=\exp\{-\varepsilon \ell[s, \tilde X(\tau),\lambda(s),m(s)]\}\int_{\mathbb{R}} [\tilde X(\tau)+\tilde Y] \exp\left[-\varepsilon \left[\tilde a\tilde Y^2+\tilde b\tilde Y\right]\right] d\tilde Y\notag\\
	&=\exp\left\{\varepsilon \left[\frac{\tilde b^2}{4\tilde a^2}-\ell[s, \tilde X(\tau),\lambda(s),m(s)]\right]\right\} \biggr[\tilde X(\tau)\sqrt{\frac{\pi}{\varepsilon \tilde a}}\notag\\
	&\hspace{1cm}+\int_{\mathbb{R}} \tilde Y \exp\left\{-\varepsilon \left[\tilde a\left(\tilde Y+\frac{\tilde b}{2\tilde a}\right)^2\right]\right\} d\tilde Y\biggr].
	\end{align}

    Substituting $k=\tilde Y+\tilde b/(2\tilde a)$ in Equation~(\ref{action18}) yields,
\begin{align}\label{action19}
	&\exp\left\{\varepsilon \left[\frac{\tilde b^2}{4\tilde a^2}-\ell[s, \tilde X(\tau),\lambda(s),m(s)]\right]\right\} \biggr[\tilde X(\tau)\sqrt{\frac{\pi}{\varepsilon \tilde a}}+\int_{\mathbb{R}} \left(k-\frac{\tilde b}{2\tilde a}\right) \exp[-\tilde a\varepsilon k^2] dk\biggr]\notag\\
	&=\exp\left\{\varepsilon \left[\frac{\tilde b^2}{4\tilde a^2}-\ell[s, \tilde X(\tau),\lambda(s),m(s)]\right]\right\} \biggr[\tilde X(\tau)-\frac{\tilde b}{2\tilde a}\biggr]\sqrt{\frac{\pi}{\varepsilon\tilde a}}.
	\end{align}
	Hence,
\begin{align}\label{action20}
	&\frac{1}{L_\varepsilon}\frac{\partial \Psi_{s}(\tilde X)}{\partial \tilde X}\int_{\mathbb{R}} \tilde\xi \exp\left[-\varepsilon \ell\left[s,\tilde\xi,\lambda(s),m(s)\right]\right] d\tilde\xi\notag\\
	&=\frac{1}{L_\varepsilon}\frac{\partial \Psi_{s}(\tilde X)}{\partial \tilde X} \exp\left\{\varepsilon \left[\frac{\tilde b^2}{4\tilde a^2}-\ell[s, \tilde X(\tau),\lambda(s),m(s)]\right]\right\} \biggr[\tilde X(\tau)-\frac{\tilde b}{2\tilde a}\biggr]\sqrt{\frac{\pi}{\varepsilon \tilde a}}.
	\end{align}

    Plugging in Equations~(\ref{action17}) and~(\ref{action20})  into Equation~(\ref{action13}) yields,
\begin{align}\label{action24}
	&\Psi_{s}(\tilde X)+\varepsilon \frac{\partial \Psi_{s}(\tilde X)}{\partial s}\notag\\
	&=\frac{1}{L_\varepsilon} \sqrt{\frac{\pi}{\varepsilon \tilde a}}\Psi_{s}(\tilde X) \exp\left\{\varepsilon \left[\frac{\tilde b^2}{4\tilde a^2}-\ell[s, \tilde X(\tau),\lambda(s),m(s)]\right]\right\}\notag\\
	&+\frac{1}{L_\varepsilon}\frac{\partial \Psi_{s}(\tilde X)}{\partial \tilde X} \sqrt{\frac{\pi}{\varepsilon \tilde a}} \exp\left\{\varepsilon \left[\frac{\tilde b^2}{4\tilde a^2}-\ell[s, \tilde X(\tau),\lambda(s),m(s)]\right]\right\} \biggr[\tilde X(\tau)-\frac{\tilde b}{2\tilde a}\biggr]+o(\varepsilon^{1/2}).
	\end{align}
	Since $\ell$ be in Schwartz space,the derivatives  are rapidly falling.  Moreover, assuming $0<|\tilde b|\leq\eta\varepsilon$, $0<|\tilde a|\leq\mbox{$\frac{1}{2}$}\left[1-\tilde\xi^{-2}\right]^{-1}$ and $\tilde X(s)-\tilde X(\tau)=\tilde \xi$ yields,
	\begin{align*}
	\tilde X(\tau)-\frac{\tilde b}{2\tilde a}=\tilde X(s)-\tilde\xi-\frac{\tilde b}{2\tilde a}=\tilde X(s)-\frac{\tilde b}{2\tilde a},\ \forall\ \tilde\xi\downarrow 0,
	\end{align*}
	such that 
	\begin{align*}
	\bigg|\tilde X(s)-\frac{\tilde b}{2\tilde a}\bigg|=\biggr|\frac{\eta\varepsilon}{\tilde\xi^2}-\eta\varepsilon\left[1-\frac{1}{\tilde\xi^2}\right]\biggr|\leq\eta\varepsilon.
	\end{align*}

    Therefore, the~Wick rotated Schr\"odinger-type Equation for the atomistic firm is,
\begin{align}\label{action25.4}
	\frac{\partial \Psi_{s}(\tilde X)}{\partial s}&=\left[\frac{\tilde b^2}{4\tilde a^2}-\ell[s, \tilde X(\tau),\lambda(s),m(s)]\right]\Psi_{s}(\tilde X).
	\end{align}
	Differentiating  Equation~(\ref{action25.4}) with respect to $m$ yields
\begin{align}\label{w18}
	\left\{\frac{2\ell_{\tilde X}}{\ell_{\tilde X\tilde X}}\left[\frac{\ell_{\tilde X\tilde X}\ell_{\tilde Xm}-\ell_{\tilde X}\ell_{\tilde X\tilde X m}}{(\ell_{\tilde X\tilde X})^2}\right]-\ell_m\right\}\Psi_{s}(\tilde X)=0,
	\end{align}
	where $\ell_{\tilde X}=\frac{\partial}{\partial \tilde X} \ell$, $\ell_{\tilde X\tilde X}=\frac{\partial^2}{\partial (\tilde X)^2} \ell$, $\ell_{\tilde X m}=\frac{\partial^2}{\partial \tilde X\partial m} \ell$ and $\ell_{\tilde X\tilde X m}=\frac{\partial^3}{\partial (\tilde X)^2\partial m} \ell=0$. Thus, optimal complementarity strategy of a player due to random monetary shock in stochastic markup dynamics is represented as $m^{*}(s,\tilde X)$ and is found by setting Equation~(\ref{w18}) equal to zero. Hence, $m^{*}(s,\tilde X)$ is the solution of the following Equation
\begin{align}\label{w21}
	\ell_u (\ell_{\tilde X\tilde X})^2=2\ell_{\tilde X} \ell_{\tilde X m}.
	\end{align}
	This completes the proof. $\square$
	
\subsection*{Availability of data}
Together with the CPI data, the stock price data for  Nestlé, Palmolive, Westrock Coffee, and Dover are used for 2022 and 2023.
\subsection*{Competing interests}
No potential conflict of interest was reported by the authors.	
\subsection*{Funding}
No funding has been used to write this paper.

	\bibliographystyle{apalike}
	\bibliography{bib}
\end{document}